\newif\ifsingle

\ifsingle
\documentclass[12pt,draftcls, onecolumn]{IEEEtran}        
\else        
\documentclass[10pt,final, twocolumn]{IEEEtran}
\fi

\usepackage{times}
\usepackage{amsmath,dsfont}
\usepackage{amssymb,amsthm}
\usepackage{epsfig,verbatim}
\usepackage{subfigure}
\usepackage{setspace}
\usepackage{color}
\usepackage{cite}
\usepackage{epstopdf}
\usepackage{graphics}
\usepackage{acronym}
\usepackage[bookmarks,colorlinks]{hyperref}
\usepackage{booktabs}
\usepackage{mathtools}
\usepackage{algorithm}
\usepackage{algorithmic}

\newtheorem{thm}{Theorem}
\newtheorem{corollary}{Corollary}
\newtheorem{prop}{Proposition}
\newtheorem{lem}{Lemma}

\acrodef{bs}[BS]{base station}
\acrodef{mimo}[MIMO]{multiple-input multiple-output}
\acrodef{mac}[MAC]{multiple access channel}
\acrodef{tdd}[TDD]{time-division duplex}
\acrodef{ut}[UT]{user terminal}
\acrodef{sinr}[SINR]{signal-to-interference-and-noise ratio}
\acrodef{mmse}[MMSE]{minimum mean-squared error}
\acrodef{mse}[MSE]{mean-squared error}
\acrodef{psd}[PSD]{positive semi-defninte}
\acrodef{svd}[SVD]{signular value decomposition}
\acrodef{GRTM}[GRTM]{greedy ratio trace maximization}
\acrodef{GSRTM}[GSRTM]{greedy sum of ratio traces maximization}
\acrodef{magiq}[MaGiQ]{Minimal gap iterative quantization} 

\newcommand{\Nrf}{N_{\rm RF}}                            
\newcommand{\Nbs}{N_{\rm BS}}                            
\newcommand{\Tpilots}{\tau}                                
\newcommand{\Wrf}{\mathcal{W}_{\rm RF}}                                
\newcommand{\Sk}[1]{\boldsymbol{S}_{\left(#1\right)}}
\newcommand{\Ind}{1\leq i\leq M}                                

\IEEEoverridecommandlockouts

\title{Pilot Contamination Mitigation with Reduced RF Chains}

\author{Shahar Stein Ioushua, \emph{Student IEEE} and Yonina C. Eldar, \emph{Fellow IEEE} 
	\thanks{Submitted for review on January 16, 2018.
		This project has received funding from the European Union's Horizon 2020 research and innovation program under grant agreement No. 646804-ERC-COG-BNYQ, and from the Israel Science Foundation under Grant No. 335/14.}
	\thanks{S. S. Ioushua is with the Dept. of Electrical Engineering, Technion - Israel Institute of Technology (e-mail: shahar-stein@campus.technion.ac.il).}    
	\thanks{Y. C. Eldar is with the Dept. of Electrical Engineering, Technion - Israel Institute of Technology (e-mail: yonina@ee.technion.ac.il).}}

\begin{document}
	\maketitle
	
	\begin{abstract}
		Massive multiple-input multiple-output (MIMO) communication is a promising technology for increasing spectral efficiency in wireless networks.
		Two of the main challenges massive MIMO systems face are degraded channel estimation accuracy due to pilot contamination and increase in computational load and hardware complexity due to the massive amount of antennas. 
		In this paper, we focus on the problem of channel estimation in massive MIMO systems, while addressing these two challenges: 
		We {\em jointly} design the pilot sequences to mitigate the effect of pilot contamination and propose an analog combiner which maps the high number of sensors to a low number of RF chains, thus reducing the computational and hardware cost. 
		We consider a statistical model in which the channel covariance obeys a Kronecker structure. 
		In particular, we treat two such cases, corresponding to fully- and partially-separable correlations.
		We prove that with these models, the analog combiner design can be done independently of the pilot sequences.
		Given the resulting combiner, we derive a closed-form expression for the optimal pilot sequences in the fully-separable case and suggest a greedy sum of ratio traces maximization (GSRTM) method for designing sub-optimal pilots in the partially-separable scenario.
		We demonstrate via simulations that our pilot design framework achieves lower mean squared error than the common pilot allocation framework previously considered for pilot contamination mitigation.

	\end{abstract}
	
	\vspace{-0.5cm}
	\section{Introduction}
	\label{sec:Intro}
	\vspace{-0.15cm}
	Massive \ac{mimo} wireless systems have emerged as a leading candidate
	for 5G wireless access \cite{andrews_what_2014,rusek_scaling_2013},
	offering increased throughput, which is scalable with the number of antennas.
	However, to fully utilize the possibilities of such large-scale arrays, accurate
	channel state information is crucial, making the channel estimation
	task a critical step in the communication process. 
	
	Channel estimation is typically performed by sending known pilot sequences from the \acp{ut} to the
	\ac{bs}, where different \acp{ut} are assigned
	orthogonal pilots to avoid intra-cell interference. Since the number
	of pilot symbols is limited by the coherence duration of the channel and the
	desire to avoid high training overhead, a limited number of orthogonal
	pilots can be allocated \cite{marzetta_noncooperative_2010}. Consequently, orthogonal pilots are only assigned
	to users in the same cell and are typically reused between different
	cells. Pilot reuse causes inter-cell interference referred to
	as pilot contamination \cite{marzetta_noncooperative_2010,hoydis_massive_2013},
	which is believed to be one of the main performance bottlenecks of
	massive \ac{mimo} systems.
	
	Mitigating the pilot contamination effect has been the focus of considerable research attention, see, e.g., \cite{jose_pilot_2011,ashikhmin_pilot_2012,ngo_evd-based_2012,muller_blind_2014,yin_coordinated_2013,chen_pilot_2016,su_fractional_2015,yan_pilot_2015,zhu_smart_2015,fernandes_inter-cell_2013,lu2014overview}.
	The majority of works on this subject can be divided into four main families: blind channel estimation, downlink precoding, pilot allocation and pilot design.
	Blind channel estimation aims at avoiding the need for pilot information, and instead estimating the channel from the received unknown data.
	Two such techniques were suggested in \cite{ngo_evd-based_2012,muller_blind_2014}, based on subspace projection and on eigenvalue decomposition.
	Both these methods exploit the asymptotic orthogonality between the channels of different \acp{ut} in massive \ac{mimo} systems to mitigate contamination of the estimated channels. 
	This approach does not aim at optimizing the pilot sequences.
	
	Downlink precoding methods assume fixed pilot sequences, and optimize the precoding matrix. 
	Specifically, these methods first carry out standard channel estimation with full pilot reuse, resulting in pilot contamination. Then, using the contaminated estimated channel, they optimize the downlink precoding matrices in all cells to compensate for the pilot contamination effect.
	In \cite{jose_pilot_2011}, a joint downlink precoding matrix is designed by minimizing the \ac{mse}. 
	The objective consists of two parts: the \ac{mse} of each \ac{ut} in the current cell, and the interference of the \acp{ut} in all other cells. These quantities depend on all the estimated channels in the system, thus requiring the \acp{bs} to exchange large amounts of information, which results in a large overhead to the estimation process. 
	
	An alternative precoding method is suggested in \cite{ashikhmin_pilot_2012} which exploits the contamination effect, by letting the \acp{bs} serve \acp{ut} from different cells. Each BS linearly combines the downlink data designated for all the UTs that share the same pilot sequence, thus leveraging the contaminated channel.  
	In this case, only the downlink data of all users and the channels' long-term statistics need to be shared between different cells. 
	Downlink precoding techniques optimize only the precoding matrices and not the pilot sequence, and can be applied with any given sequences. 
	In particular, they may be applied after a pilot optimization step, which, as we show in this work, contributes significantly to reducing the overall contamination.
	
	In contrast to the previous two approaches, pilot allocation methods optimize the pilot sequences themselves. 
	In this framework, the same predefined set of potential pilot sequences is used in all cells. The optimization is over the allocation of the sequences to the users, such that the cross interference caused by users that share the same sequence is minimized. 
	A common approach for pilot allocation is based on angle-of-arrival (AOA) of the \acp{ut}' signal to the \acp{bs}. One such example is  \cite{yin_coordinated_2013}, which aims to exploit the spatial orthogonality between different users and proposes a greedy pilot assignment
	algorithm, which minimizes the sum of estimation errors in multipath channels by assigning the same pilot sequence to \acp{ut} with non-overlapping angle spread. 
	Later results in \cite{khormuji_pilot-decontamination_2016,saxena_mitigating_2015} suggest that this method achieves the same performance as a simple random pilot assignment.
	Partial pilot reuse is proposed in \cite{su_fractional_2015,yan_pilot_2015},
	assigning non-orthogonal pilots only to cell-centered \acp{ut}, while
	\acp{ut} located at the edge of the cell are assigned orthogonal ones.
	The pilot allocation scheme in \cite{zhu_smart_2015} is based on an iterative algorithm across the different cells, 
	assigning the pilot sequence with the weakest interference to the \ac{ut} with the most attenuated channel in each iteration. 
	However, this method is only suitable when the number of \acp{ut} in each cell equals the number of pilot symbols.
	In \cite{fernandes_inter-cell_2013} the authors propose a transmission
	scheme based on time-shifted pilots, in which the cells are partitioned into groups and the scheduling
	of pilots is shifted in time from one group to the next, canceling
	interference between adjacent cells. 
	All the above works {\em assume given pilot sequences}, and do not attempt to {\em design} the sequences, but limit the optimization to their allocation.
	
	Previous pilot sequence design methods typically consider a single cell perspective  \cite{pang_optimal_2007,liu_training_2007,kotecha_transmit_2004,noh_pilot_2014,bogale_pilot_2014,bjornson_framework_2010}. 
	These works treat a single desired channel with multiple interferers and optimize the pilot sequences in the specific cell to minimize its channel estimation error assuming all the pilots of surrounding cells are given.
	Generalization of these solutions to multiple cells for jointly designing all pilot matrices in the system jointly is not trivial.
	A possible straight forward approach for generalizing these methods is to iteratively apply the single-cell solution over each cell independently. However, there is no guarantee that this extension improves the overall channel estimation performance, as the optimization focuses only a single cell. 
	In contrast, when jointly optimizing over the entire network, we show that under our channel model the solution for the joint-design problem is either obtained in closed-form or is based on a greedy algorithm, hence it is guaranteed to improve the overall channel estimation accuracy.

	None of the previous works on pilot contamination consider the problem of {\em jointly designing the pilot sequences} across the cells.
	Furthermore, previous works on pilot contamination mitigation assume dedicated RF chains per antenna. 
	As such architecture is too costly in terms of hardware, and the massive amount of data received from hundreds of antennas at each \ac{bs} constitutes a huge computational challenge, 
	it is desirable to reduce the amount of RF chains via analog combining while maintaining high accuracy in estimating the channel. 
	Previous works on analog combiner design \cite{alkhateeb_channel_2014,ayach_spatially_2014,li_hybrid_2017,mendez-rial_hybrid_2016,yu_alternating_2016,yu_partially-connected_2017,kim_mse-based_2015,ioushua_eldar_hybrid_2017} 
	focused on the problems of channel estimation and hybrid beamforming, under various hardware constraints. However, to the best of our knowledge, despite its practical importance, {\em joint analog combiner and pilot sequence design} has not been studied to date.
	
	In this work, we treat the joint design of pilot sequences and analog combiners, aimed at minimizing the sum of \acp{mse} across all cells, for a fixed number of RF chains. 
	We consider a general statistical channel model in which the channel's correlation matrices have a Kronecker structure and focus on two scenarios: fully-separable and partially-separable correlations. 
	For both cases, the only information assumed to be shared between different cells is the long-term statistics of the channels.
	We prove that, in both scenarios, the analog combiners can be first designed independently of the pilot sequences by applying existing algorithms such as the ones in \cite{ioushua_eldar_hybrid_2017}. Given the resulting combiners, the pilot sequences are designed,
	where the solution depends on the channel model.
	
	In the fully-separable case \cite{kermoal_stochastic_2002}, the optimal pilot sequences can be obtained in closed-form.
	In particular, when the transmit correlation matrix is diagonal, the received pilots correspond to user selection, i.e., choosing the \acp{ut} with strongest average links to their \ac{bs}. 
	For the partially-separable correlation model, we express the pilot sequence design problem as a maximization of a weighted sum of ratio traces. We then  propose a greedy algorithm that generalizes a previously suggested method for maximizing a single ratio trace, the \ac{GRTM} \cite{ioushua_eldar_hybrid_2017}, to a sum. 
	At each iteration of the \ac{GSRTM}, one additional pilot symbol is enabled, and the algorithm chooses the optimal symbol to add for each user, given the previous selections.
	
	We demonstrate the advantage of joint pilot design in simulations for both cases. We compare the suggested algorithms with other design methods and state-of-the-art pilot allocation techniques, and show that our algorithms enjoy lower sum-\ac{mse}.
	
	The rest of this paper is organized as follows: 
	Section~\ref{sec:model} presents the system model and problem formulation. 
	Section~\ref{sec:MMSEestimation} derives the \ac{mmse} channel estimation problem.
	In Section~\ref{sec:AnalogCombiner} the analog combiner design problem is studied. 
	Pilot sequence optimization is treated in Section~\ref{sec:PilotDesign}.  
	Section~\ref{sec:Esperiments} illustrates the performance of our design in simulation examples.

	Throughout the paper the following notations are used:
	We denote column vectors with boldface letters, e.g., ${\bf {x}}$, and matrices with boldface upper-case letters,  e.g., $\boldsymbol{X}$.
	The set of complex numbers is denoted by $\mathbb{C}$, $\mathcal{CN}$ is the complex-normal distribution, $\boldsymbol{1}_n$ is the $n \times 1$ all ones vector, and $\boldsymbol{I}_n$ is the $n\times n$ identity matrix.  
	Hermitian transpose, transpose, Frobenius norm, and Kronecker product  are denoted by $(\cdot)^*$, $(\cdot)^T$, $\|\cdot\|_F$, and  $\otimes$, respectively. 
	For two $n \times 1$ real-valued vectors ${\bf x}_1, {\bf x}_2$, the inequality ${\bf x}_1 \le {\bf x}_2$ indicates that all the entries of ${\bf x}_1$ are less than or equal to the corresponding entries of ${\bf x}_2$. 
	For an $n \times n$ matrix $\boldsymbol{X}$, ${\rm {tr}}\left(\boldsymbol{X}\right)$ is the trace of $\boldsymbol{X}$, 
	$\lambda_{i}\left(\boldsymbol{X}\right)$ is the $i$-th largest real eigenvalue of $\boldsymbol{X}$ and
	$\mbox{vec}\left(\boldsymbol{X}\right)$ is the $n^2 \! \times \! 1$ column vector obtained by stacking the columns of $\boldsymbol{X}$ one below the other.
	The $n \times n$ diagonal matrix $\mbox{diag}\left( \boldsymbol{x}\right)$ has the vector $\boldsymbol{x}$ on its diagonal,
	and $\mbox{diag}^{-1}\left( \boldsymbol{X}\right)$ is the $n\times 1$ vector whose entries are the diagonal elements of $\boldsymbol{X}$. 
	Finally, for a sequence of $n \times m$ matrices $\{\boldsymbol{X}_i\}_{i=1}^{k}$, $\mbox{blkdiag}\left( \boldsymbol{X}_1,\ldots,\boldsymbol{X}_k\right)$ is the $kn \times km$ block-diagonal matrix with on-diagonal matrices $\{\boldsymbol{X}_i\}_{i=1}^{k}$.

	\vspace{-0.1cm}
	\section{Problem Formulation and Channel Model}
	\label{sec:model}
	\subsection{Problem Formulation}
	Consider a network of $M$ time-synchronized cells with full spectrum reuse. 
	In each cell, a \ac{bs}   equipped with $\Nbs$ antennas and $\Nrf\leq \Nbs$ RF chains serves $K$ single antenna \acp{ut}. 
	At the \ac{bs}, a network of analog components, typically phase shifters and switches, maps the $\Nbs$ antennas to the $\Nrf$ RF chains. 
	In the digital domain, only the reduced number of inputs, $\Nrf$, are accessible.
	
	In the uplink channel estimation phase, each \ac{ut} sends $\Tpilots<MK$ training symbols to the \ac{bs}, during which the channel is assumed to be constant. 
	Let $\boldsymbol{s}_{ik}$  be the $\Tpilots \times 1$ pilot sequence vector of the $k$-th user in the $i$-th cell, assumed to be subject to a  per-user power constraint  $\boldsymbol{s}_{ik}^{*}\boldsymbol{s}_{ik}\leq\mathcal{P}$.  
	The pilot sequence matrix of the \acp{ut} in the $i$-th cell is denoted by $\boldsymbol{S}_{i} = \big[\boldsymbol{s}_{i1}, \ldots,  \boldsymbol{s}_{iK}\big]$.
	Let $\boldsymbol{H}_{ij}\in\mathbb{C}^{\Nbs\times K}$  be the channel matrix between the \acp{ut} of the $j$-th cell and the $i$-th \ac{bs},
	and $\boldsymbol{W}_{ij}\in\Wrf$ be the analog combiner matrix at the $i$-th \ac{bs}, where $\Wrf \subseteq \mathbb{C}^{\Nbs\times \Nrf}$ is the feasible set of analog matrices. The properties of $\Wrf$ are discussed below.

	We consider an interference limited case, based on the analysis in \cite{marzetta_noncooperative_2010}, which shows that in the limit of an infinite number of \ac{bs} antennas, the effect of uncorrelated noise vanishes. 
	Thus, the (discrete-time) received signal at the $i$-th cell \ac{bs}, $\boldsymbol{Y}_{i} \in \mathbb{C}^{\Nbs\times \Tpilots}$, $1\leq i\leq M$, is given by
	\begin{equation}
	\boldsymbol{Y}_{i}  =\boldsymbol{W}_{i}\boldsymbol{H}_{ii}\boldsymbol{S}_{i}^{T}+\boldsymbol{W}_{i}\sum_{j\neq i}\boldsymbol{H}_{ij}\boldsymbol{S}_{j}^{T}.\label{eq:Receive Signal}
	\end{equation}
	At each \ac{bs} the channel $\boldsymbol{H}_{ii}$ is estimated from the measurements $\boldsymbol{Y}_{i}$, using an \ac{mmse} estimator, where the  expectation is taken over the channels $\boldsymbol{H}_{ij}$, $\Ind$.
	We note that the $i$th \ac{bs}, $1\leq i\leq M$, needs to estimate only the channel to its corresponding \acp{ut}, namely, $\boldsymbol{H}_{ii}$.
	
	The pilot contamination problem implies that the estimation of $\boldsymbol{H}_{ii}$ is degraded by the presence of the interfering channels $\boldsymbol{H}_{ij}$, $i\neq j$, in the received signal $\boldsymbol{Y}_{i}$. Since the dimension $\tau$ of the pilot matrices, $\left\{\boldsymbol{S}_i\right\}_{i=1}^M$ is smaller than the total number of users in the system $MK$, these interferences can not be fully eliminated at the $i$th \ac{bs}. 
	We wish to jointly design the pilot sequences and analog combiners $\left\{ \boldsymbol{S}_{i},\boldsymbol{W}_{i}\right\} _{i=1}^{M}$ to minimize the sum of \acp{mse} of the desired channels $\boldsymbol{H}_{ii}$ from the observations $\boldsymbol{Y}_{i}$, $1\leq i\leq M$, over the entire network, subject to the per-user power constraint on $\boldsymbol{S}_{i}$, and to the feasible set of $\boldsymbol{W}_{i}$. 
	
	The feasible set $\Wrf$ depends on the specific architecture of the analog network, which can vary according to a different power, area, and budget constraints. 
	For example, in the case of a fully-connected phase shifters network, $\Wrf$ is the set of $\Nrf \times \Nbs$ unimodular matrices.
	For a review of common hardware schemes, the reader is referred to \cite{mendez-rial_hybrid_2016,ioushua_eldar_hybrid_2017}.
	
	Designing pilot sequences in the interference-limited setup \eqref{eq:Receive Signal} without RF reduction has been previously considered by \cite{bjornson_framework_2010,kotecha_transmit_2004} under a fully-separable Kronecker model. 
	These works treated a single desired channel, $\boldsymbol{H}_{11}$, with multiple interferers, $\boldsymbol{H}_{1j}$, $j=2,\cdots,M$, that share the same receive side correlation. 
	In contrast to our approach, they optimized only the pilot sequence of the specific cell, $\boldsymbol{S}_{1}$, to minimize its channel estimation error assuming all the interferers pilot matrices $\boldsymbol{S}_{j}$, $j=2,\cdots,M$, are given.
	This leads to a water filling solution, assigning more power to directions with larger channel gains and weaker interference. 
	Here, we aim at jointly optimizing the pilot matrices of all cells. Furthermore, we consider both the fully- and partially-separable correlation structures for the channel's covariance, rather than just the first.
	
	\subsection{Channel Model}
	We focus on statistical channel models obeying the following structure:
	\begin{equation}
	\boldsymbol{H}_{ij}=\boldsymbol{Q}_{ij}^{\frac{1}{2}}\boldsymbol{\bar{H}}_{ij}\boldsymbol{P}_{ij}^{\frac{1}{2}}.\label{eq:GeneralModel}
	\end{equation}
	Here $\boldsymbol{Q}_{ij}\in\mathbb{C}^{\Nbs\times \Nbs},\boldsymbol{P}_{ij}\in\mathbb{C}^{K\times K}$ are deterministic \ac{psd} Hermitian matrices, referred to as the correlation matrices, and  $\boldsymbol{\bar{H}}_{ij}\in\mathbb{C}^{\Nbs\times K},  1\leq i,j\leq M$
	are random matrices with i.i.d complex-normal zero-mean unit variance entries, representing the fast-fading channel coefficients. 
	We assume that $\boldsymbol{Q}_{ij},\boldsymbol{P}_{ij}$, $1\leq i,j\leq M$, are known at all the \acp{bs}. 
	We focus on two special cases of (\ref{eq:GeneralModel}):
	\begin{enumerate}
		\item \textbf{Fully-Separable Correlations.}
		In this case,  
		\begin{equation}
		\begin{aligned}
		\boldsymbol{Q}_{ij}=\boldsymbol{Q}_{i},\;
		\boldsymbol{P}_{ij}=\boldsymbol{P}_{j},\label{eq:FullySeparable}
		\end{aligned}
		\end{equation}
		i.e., the receive and transmit side correlations are independent of each other.
		This model is also known as the \textit{doubly correlated Kronecker model} \cite{tulino_impact_2005}.
		Here $\boldsymbol{Q}_{i}=\boldsymbol{R}_{r_{i}}$ is the receive side correlation matrix of the $i$-th cell (assumed here for simplicity to be full rank), 
		and $\boldsymbol{P}_{j}=\boldsymbol{R}_{t_{j}}$ the transmit side correlation matrix of the users in the $j$-th cell.
		It is a reasonable assumption that $\boldsymbol{P}_{j}$ is an invertible diagonal matrix due to the distance between different \acp{ut}. However, in our setup we allow general transmit correlation matrices.
		
		The fully-separable model implies that the transmitters do not affect the spatial properties of the received signal and vice versa.
		This model was shown to accurately model systems using space diversity arrays \cite{tulino_impact_2005} and to fit various wireless scenarios, such as indoor MIMO environments \cite{Yu01secondorder}. 
		Nonetheless, we note that it does not account for the distance of different users from the same BS, and it inherently rules out spatial orthogonality between different \acp{ut}.
		\label{model:Kronecker}
		
		\smallskip
		\item \textbf{Partially Separable Correlations.} 
		In this model only $\boldsymbol{Q}_{ij}$ is separable, that is
		\begin{equation}
		\begin{aligned}
		\boldsymbol{Q}_{ij}=\boldsymbol{Q}_{i}.\label{eq:PartiallSeparable}
		\end{aligned}
		\end{equation}
		
		An example of a partially-separable channel is the widely used model of \cite{marzetta_noncooperative_2010}, which hereinafter we refer to as the {\em MU-MIMO channel fading model}. In this case, the columns of $\boldsymbol{H}_{ij}$, denoted $\{\boldsymbol{g}_{ikj}\}_{k=1}^{K}$,  are independent random vectors distributed as $\boldsymbol{g}_{ikj}\sim\mathcal{CN}\left(0,\beta_{ikj}\boldsymbol{I}_{\Nbs}\right)$ for some set of decay factors $\{\beta_{ikj}\}$. The channel matrices can then be written as     
		\begin{equation}
		\boldsymbol{H}_{ij}=\boldsymbol{\bar{H}}_{ij}\boldsymbol{D}_{ij}^{\frac{1}{2}}, \label{eq:MarzetaModel}
		\end{equation}
		with $\boldsymbol{D}_{ij}=\text{diag}\left(\beta_{i1j}\cdots\beta_{iKj}\right)$,
		resulting in 
		\begin{equation}
		\begin{aligned}
		\boldsymbol{Q}_{i}=\boldsymbol{I}_{\Nbs}, \; 
		\boldsymbol{P}_{ij}=\boldsymbol{D}_{ij}.\label{MarzettaCorrelations}
		\end{aligned}
		\end{equation} 
	\end{enumerate}
	
	In this work, the correlation matrices $\{\boldsymbol{Q}_{ij},\boldsymbol{P}_{ij}\}$ are assumed to be a-priori known. 
	In practice, they need to be estimated, which can be computationally complex in massive \ac{mimo} systems due to the large number of estimated matrix entries. 
	However, under the Kronecker structures considered here, the number of parameters to be estimated is significantly reduced. 
	Moreover, the correlation matrices are typically long-term statistics, varying much slower compared to the fast-fading channel coefficients, so that estimation can be performed once every many coherence intervals, without substantial overhead.

	\vspace{-0.2cm}
	\section{MMSE Channel Estimation}
	\label{sec:MMSEestimation}
	\vspace{-0.1cm}
	We begin by deriving the \ac{mmse} channel estimator and its resulting
	\ac{mse} under the model (\ref{eq:GeneralModel}). 
	By vectorizing the received signal in (\ref{eq:Receive Signal}), $\boldsymbol{y}_{i} = {\rm vec}\left( \boldsymbol{Y}_{i}\right)$,  we
	have
	\begin{equation}
	\boldsymbol{y}_{i} =\left(\boldsymbol{S}_{i}\otimes \boldsymbol{W}_{i}\right)\boldsymbol{h}_{ii}+\sum_{j\neq i}\left(\boldsymbol{S}_{j}\otimes \boldsymbol{W}_{i}\right)\boldsymbol{h}_{ij},\label{eq:vectorized}
	\end{equation}
	where $\boldsymbol{h}_{ij}=\mbox{vec}\left(\boldsymbol{H}_{ij}\right)$ and
	$\boldsymbol{h}_{ij}\sim\mathcal{CN}\left(\boldsymbol{0},\boldsymbol{P}_{ij}\otimes \boldsymbol{Q}_{ij}\right)$.
	The $i$-th \ac{bs} estimates its desired channel $\boldsymbol{h}_{ii}$ 
	using the \ac{mmse} estimator \cite[Chp. 12]{Kay:1993:FSS:151045}, given by 
	\begin{equation}
	\hspace{-0.3cm}\hat{\boldsymbol{h}}_{ii} \!=\!\left(\boldsymbol{C}_{ii}\boldsymbol{S}_{i}^{*}\!\otimes\!\boldsymbol{Q}_{ii}\boldsymbol{W}_{i}^{*}\right)\!\!\left[\sum_{j=1}^{M}\!\left(\boldsymbol{S}_{j}\boldsymbol{P}_{ij}\boldsymbol{S}_{j}^{*}\!\otimes\!\boldsymbol{W}_{i}\boldsymbol{Q}_{ij}\boldsymbol{W}_{i}^{*}\right)\right]^{-1}\hspace{-0.4cm}.\label{eq:LMMSEestimatorGeneral}
	\end{equation}
	Here, we assume all inverses exist, but similar derivation can be made using the pseudo-inverse.
	The corresponding estimation error for \eqref{eq:LMMSEestimatorGeneral} is given by 
	\begin{align}
	\epsilon_{i} & =\mbox{tr}\left(\boldsymbol{A}_{i}\right)-\mbox{tr}\left(\boldsymbol{B}_{i}\boldsymbol{C}_{i}\boldsymbol{B}_{i}^{*}\right),\label{eq:MMSE_ErrorGeneral}
	\end{align}
	with $\boldsymbol{A}_{i}\triangleq\boldsymbol{P}_{ii}\otimes\boldsymbol{Q}_{ii}$,
	\begin{equation}\label{eq:Bi}
	\boldsymbol{B}_{i}\triangleq\boldsymbol{P}_{ii}\boldsymbol{S}_{i}^{*}\otimes\boldsymbol{Q}_{ii}\boldsymbol{W}_{i}^{*},
	\end{equation}
	and 
	\begin{equation}\label{eq:Ci}
	\boldsymbol{C}_{i}\triangleq\left[\sum_{j=1}^{M}\boldsymbol{S}_{j}\boldsymbol{P}_{ij}\boldsymbol{S}_{j}^{*}\otimes\boldsymbol{W}_{i}\boldsymbol{Q}_{ij}\boldsymbol{W}_{i}^{*}\right]^{-1}.
	\end{equation}
	
	We wish to design the pilot sequence matrices $\boldsymbol{S}_{i}$ and analog
	combining matrices $\boldsymbol{W}_{i}$ for $1\leq i\leq M$ to minimize the sum
	of the errors given by
	\begin{align}
	\epsilon =\sum_{i=1}^{M}\epsilon_i  =\sum_{i=1}^{M}\left(\mbox{tr}\left(\boldsymbol{A}_{i}\right)-\mbox{tr}\left(\boldsymbol{B}_{i}\boldsymbol{C}_{i}\boldsymbol{B}_{i}^{*}\right)\right),\label{eq:sum_of_errors}
	\end{align}
	under the given power and hardware constraints. 
	This approach will result in a different solution than the single target cell approach in  \cite{pang_optimal_2007,liu_training_2007,kotecha_transmit_2004,noh_pilot_2014,bogale_pilot_2014,bjornson_framework_2010}, as an interfering user in one cell
	is the desired user in another. 
	
	Since the matrices $\boldsymbol{A}_{i}$ do not depend on the optimization variables
	$\boldsymbol{W}_{i}$ and $\boldsymbol{S}_{i}$, minimizing $\epsilon$ is equivalent
	to maximization of $\sum_{i}\mbox{tr}\left(\boldsymbol{B}_{i}\boldsymbol{C}_{i}\boldsymbol{B}_{i}^{*}\right)$, which yields the following optimization problem:
	%
	\begin{align}
	\label{eq:GeneralOptimization}
	\underset{_{{\boldsymbol{S}_{i},\boldsymbol{W}_{i}}}}{\max} & \;\sum_{i}\mbox{tr}(\boldsymbol{B}_{i}\boldsymbol{C}_{i}\boldsymbol{B}_{i}^{*})&\\ \notag 
	\mbox{s.t.}\;\;\; & \;\boldsymbol{s}_{ik}^{*}\boldsymbol{s}_{ik}\leq\mathcal{P},\;\; 1\leq i\leq M, \;1\leq k\leq K,& \\ \notag
	& \;\boldsymbol{W}_{i}\in\Wrf,\; 1\leq i\leq M\quad & 
	\end{align}
	where $\boldsymbol{B}_{i}$ and $\boldsymbol{C}_{i}$ are given by \eqref{eq:Bi} and \eqref{eq:Ci} respectively.
	
	In the next section, we show that in both scenarios \eqref{eq:FullySeparable} and \eqref{eq:PartiallSeparable} the combiner design can be performed independently of the pilot sequences. Furthermore, previously suggested methods such as those of \cite{ioushua_eldar_hybrid_2017} can be used. In Section~\ref{sec:PilotDesign} we solve the remaining pilot design problem for each scenario separately. 
	
	\section{Analog Combiner Design}
	\label{sec:AnalogCombiner}
%
	\subsection{Analog Combiner Design Problem Formulation}
	Note that in both correlation scenarios \eqref{eq:FullySeparable} and \eqref{eq:PartiallSeparable} we have $\boldsymbol{Q}_{ij}=\boldsymbol{Q}_{i}$. 
	By defining
	
	\begin{equation}
	w_{i}  \triangleq\mbox{tr}\left(\boldsymbol{Q}_{i}\boldsymbol{W}_{i}^{*}\left(\boldsymbol{W}_{i}\boldsymbol{Q}_{i}\boldsymbol{W}_{i}^{*}\right)^{-1}\boldsymbol{W}_{i}\boldsymbol{Q}_{i}\right), \label{eq:wights}
	\end{equation}
	and using standard Kronecker product properties, \eqref{eq:GeneralOptimization} becomes  
	\begin{align}
	\label{eq:GeneralOptimization2}
	\underset{_{\boldsymbol{S}_{i},\boldsymbol{W}_{i}}}{\max} & \;\sum_{i}w_i \cdot \mbox{tr}\left(\boldsymbol{S}_{i}\boldsymbol{P}_{ii}^{2}\boldsymbol{S}_{i}^{*}\left[\sum_{j}^{M}\boldsymbol{S}_{j}\boldsymbol{P}_{ij}\boldsymbol{S}_{j}^{*}\right]^{-1}\right)&\\ \notag 
	\mbox{s.t.}\;\;\; & \;\boldsymbol{s}_{ik}^{*}\boldsymbol{s}_{ik}\leq\mathcal{P},\;\; 1\leq i\leq M, \;1\leq k\leq K,\quad \\ \notag
	& \;\boldsymbol{W}_{i}\in\Wrf,\; 1\leq i\leq M.\quad & 
	\end{align}
	
	Since the matrices $\boldsymbol{P}_{ij}$ are \ac{psd} for all $1\leq i,j\leq M$, it follows that
	\begin{equation*}
	\mbox{tr}\left(\boldsymbol{S}_{i}\boldsymbol{P}_{ii}^{2}\boldsymbol{S}_{i}^{*}\left[\sum_{j}^{M}\boldsymbol{S}_{j}\boldsymbol{P}_{ij}\boldsymbol{S}_{j}^{*}\right]^{-1}\right)\geq 0,\; 1\leq i\leq M .
	\end{equation*} 
	Consequently, the objective in \eqref{eq:GeneralOptimization2} is a weighted sum of non-negative elements, and hence an increasing function in all the weights $w_i$, $1\leq i\leq M$. 
	Since the constraint on $\boldsymbol{W}_i$ is independent of $\left\{\boldsymbol{S}\right\}_{i=1}^M$ and of $\boldsymbol{W}_j$ for $j \neq i$, we can maximize the sum in \eqref{eq:GeneralOptimization2} by maximizing each weight $w_i$ separately for each cell and independently of $\left\{\boldsymbol{S}\right\}_{i=1}^M$. The resulting $w_i$, $\Ind$, can be plugged back into \eqref{eq:GeneralOptimization2} and the problem can be solved with respect to the pilot matrices. 
	Note that since the optimization of the combiner is independent of the pilot sequences, no alternations between the two is required.
	
	The optimization problem for the analog combiner of the $i$-th \ac{bs} is therefore
	\begin{eqnarray}
	\label{eq:W_optimization}
	& \underset{_{\boldsymbol{W}_i}}{\max} & \mbox{tr}\left(\boldsymbol{W}_i\boldsymbol{Q}_{i}^{2}\boldsymbol{W}_{i}^{*}\left(\boldsymbol{W}_{i}\boldsymbol{Q}_{i}\boldsymbol{W}_{i}^{*}\right)^{-1}\right)\\ \notag
	& \mbox{s.t.}\;\; & \boldsymbol{W}_{i}\in\Wrf.
	\end{eqnarray}
	Since \eqref{eq:W_optimization} can be solved for each \ac{bs} separately, for clarity, we drop the index $i$.
	This problem has been previously studied in \cite{spawc_paper,ioushua_eldar_hybrid_2017}, under various hardware constraints.
	In \cite{spawc_paper}, a permissive analog scheme that consists of phase shifters, gain shifters and switches was considered. 
	The feasible set $\Wrf$ in this case is the set of all matrices $\boldsymbol{W}\in \mathbb{C}^{\Nrf \times \Nbs}$. 
	For this choice, the family of optimal solutions to \eqref{eq:W_optimization} is given by
	\begin{equation}
	\boldsymbol{W}_{\rm FD} = \boldsymbol{T}\tilde{\boldsymbol{U}}^{*},\label{eq:WfullyDigital}
	\end{equation}
	where $\tilde{\boldsymbol{U}}$ is an $\Nrf \times \Nbs$ matrix whose columns are the eigenvectors of the Hermitian matrix $\boldsymbol{Q}$ corresponding to its $\Nrf$ largest eigenvalues, and $\boldsymbol{T}$ is some $\Nrf\times \Nrf$ invertible matrix.
	The combiner \eqref{eq:WfullyDigital} is commonly referred to as the fully-digital solution. 
	The specific solution derived in \cite{spawc_paper} is \eqref{eq:WfullyDigital} with $\boldsymbol{T}=\boldsymbol{I}_{\Nrf}$.
	Problem \eqref{eq:W_optimization} is also a specific case of the framework treated in \cite{ioushua_eldar_hybrid_2017}. There, more restrictive hardware schemes were considered, under which the optimal solution for \eqref{eq:W_optimization} may be intractable. 
	For these scenarios, two low complexity algorithms were suggested to obtain a sub-optimal solution: \ac{magiq} and \ac{GRTM}.
	Both techniques may be used here to solve \eqref{eq:W_optimization}. 
	For completeness, in the following subsection we describe \ac{magiq}, \cite{ioushua_eldar_hybrid_2017}. 
	
	After the analog combiners are designed for each cell separately, the pilot sequences are optimized given the resulting combiners.
	The impact of the RF reduction on the channel's MSE is illustrated via simulation examples in Section~\ref{sec:Esperiments}.
	
	\subsection{MaGiQ - Minimal Gap Iterative Quantization}
	\label{sec:MaGiQ}
	
	One approach for treating \eqref{eq:W_optimization}, is to approximate the fully-digital solution \eqref{eq:WfullyDigital} with an analog one. 
	It was shown in \cite{ioushua_eldar_hybrid_2017} that minimizing the approximation gap between the analog combiner $\boldsymbol{W}$ and the fully-digital combiner $\boldsymbol{W}_{\rm FD}$ minimizes an upper bound on the MSE \eqref{eq:W_optimization}. 
	Motivated by this insight, it was suggested to consider the problem 
	\begin{eqnarray}
	\label{eq:WT_optimization}
	& \underset{_{\boldsymbol{W},\boldsymbol{T}}}{\min} & \Vert \boldsymbol{T}\tilde{\boldsymbol{U}}^{*} - \boldsymbol{W}\Vert_{F}^2\\
	& \mbox{s.t.}\;\; & \boldsymbol{W}\in\Wrf,\quad \boldsymbol{T}\in{\mathcal{U}},\notag
	\end{eqnarray}  
	where $\mathcal{U}$ is the set of unitary $\Nrf\times \Nrf$ matrices.
	Different from previous approaches \cite{mendez-rial_hybrid_2016,yu_alternating_2016} which considered $\boldsymbol{T}=\boldsymbol{I}_{\Nrf}$ and optimized \eqref{eq:WT_optimization} over $\boldsymbol{W}$ alone, here the optimization is carried out with respect to both $\boldsymbol{W}$ and $\boldsymbol{T}$. 
	To solve \eqref{eq:WT_optimization}, an alternating minimization approach was suggested, referred to as \ac{magiq}.
	
	The benefit of \ac{magiq}, is that at each iteration, a closed-form solution for both variables is available.
	For fixed $\boldsymbol{W}$, the optimal $\boldsymbol{T}$ is given by 
	\begin{equation}
	\boldsymbol{T}_{\rm opt}=\bar{\boldsymbol{V}}\bar{\boldsymbol{U}}^{*}
	\end{equation}
	with $\tilde{\boldsymbol{U}}^{*}\boldsymbol{W}^{*}=\bar{\boldsymbol{U}}\boldsymbol{\Lambda}\bar{\boldsymbol{V}}^{*}$ the \ac{svd} of $\tilde{\boldsymbol{U}}^{*}\boldsymbol{W}^{*}$.
	For a given $\boldsymbol{T}$, the optimal combiner is
	\begin{equation}
	\boldsymbol{W}_{\rm opt}=\boldsymbol{P}_{\Wrf}\left(\boldsymbol{T}\tilde{\boldsymbol{U}}^{*}\right)
	\end{equation}
	with $\boldsymbol{P}_{\mathcal{W}}\left(\boldsymbol{A}\right)$ denoting the orthogonal projection of the matrix $\boldsymbol{A}$ onto the set $\mathcal{W}$.  
	\ac{magiq} is summarized in Algorithm~\ref{algo:SimpleQuan}.
	\begin{algorithm}[h] 
		\caption{MaGiQ -  Minimal Gap Iterative Quantization}
		\text{\textbf{Input:} $\tilde{\boldsymbol{U}}^{*}$, threshold $t$}  \\  
		\text{\textbf{Output:} analog combiner $\boldsymbol{W}$}\\
		\text{\textbf{Initialize} $\boldsymbol{T}=\boldsymbol{I}_{\Nrf},\boldsymbol{W}=\boldsymbol{0}$}\\
		\text{\textbf{While} $\Vert\boldsymbol{T}\tilde{\boldsymbol{U}}^{*}-\boldsymbol{W}\Vert_{F}^{2}\geq t$ do:}
		\begin{enumerate}
			\item $\boldsymbol{W}=\boldsymbol{P}_{\mathcal{W}}(\boldsymbol{T}\tilde{\boldsymbol{U}}^{*})$\label{eq:Fstep}
			\item \text{\text{Calculate} the \ac{svd} $\tilde{\boldsymbol{U}}^{*}\boldsymbol{W}^{*}=\bar{\boldsymbol{U}}\boldsymbol{\Lambda}\bar{\boldsymbol{V}}^{*}$}
			\item \text{$\boldsymbol{T}=\bar{\boldsymbol{V}}\bar{\boldsymbol{U}}^{*}$}\label{eq:Tstep}
		\end{enumerate}
		\label{algo:SimpleQuan}
	\end{algorithm}
	
	\section{Pilot Sequence Design}
	\label{sec:PilotDesign}
	After the analog combiners are optimized, the obtained weights $w_i$ can be plugged back into \eqref{eq:GeneralOptimization2}, resulting in the following problem with respect to the pilot sequences $\boldsymbol{S}_i$, $1\leq i\leq M$:
	\begin{align}
	\label{eq:GeneralOptimization3}
	\underset{_{\boldsymbol{S}_{i}}}{\max} & \;\sum_{i}w_i \cdot \mbox{tr}\left(\boldsymbol{S}_{i}\boldsymbol{P}_{ii}^{2}\boldsymbol{S}_{i}^{*}\left[\sum_{j}^{M}\boldsymbol{S}_{j}\boldsymbol{P}_{ij}\boldsymbol{S}_{j}^{*}\right]^{-1}\right)&\\ \notag 
	\mbox{s.t.}\;\;\; & \;\boldsymbol{s}_{ik}^{*}\boldsymbol{s}_{ik}\leq\mathcal{P},\;\; 1\leq i\leq M, \;1\leq k\leq K.\quad 
	\end{align}
	We now study problem \eqref{eq:GeneralOptimization3} under the channel models \eqref{eq:FullySeparable} and \eqref{eq:PartiallSeparable}. 
	In Section~\ref{sec:FullySep}, we consider the fully-separable model in \eqref{eq:FullySeparable}, under which we derive a closed-form solution for the pilot sequences $\boldsymbol{S}_{i}$. 
	Next, in Section~\ref{sec:PartiallySep}, we focus on the partially-separable model \eqref{eq:PartiallSeparable} and express the pilot design problem as a sum of quadratic ratios. We then suggest a greedy algorithm to solve it, referred to as \ac{GSRTM}
	
	\subsection{Fully Separable Correlations}
	\label{sec:FullySep}
	For the fully-separable scenario we show that although this model may lack in description for massive MIMO systems, 
	it holds a key advantage in facilitating the solution of \eqref{eq:GeneralOptimization}, as it results in a closed form expression for the optimal pilot sequences.  
	
	Under the model \eqref{eq:FullySeparable}, we have $\boldsymbol{P}_{ij}=\boldsymbol{P}_{j}$.
	By letting 
	\begin{equation}
	\boldsymbol{S}\triangleq \left[\boldsymbol{S}_{1},\cdots,\boldsymbol{S}_{M}\right]
	\label{eq:PilotsMatrix}
	\end{equation}
	be the  $\tau\times MK$ pilot matrix, 
	\begin{equation}
	\boldsymbol{\bar{P}}\triangleq\mbox{blkdiag}\left(\left[\boldsymbol{P}_{1},\cdots,\boldsymbol{P}_{M}\right]\right),
	\label{eq:TransmitMatrix}
	\end{equation} 
	the $MK\times MK$ Hermitian transmit correlation matrix and  
	\begin{align}
	\boldsymbol{\bar{W}} & \triangleq\mbox{diag}\left(\left[w_{1},\cdots,w_{M}\right]\right)\otimes \boldsymbol{I}_{K}\label{eq:weightsMatrix}
	\end{align}
	the $MK\times MK$ non-negative valued diagonal weights matrix,
	problem \eqref{eq:GeneralOptimization3} can be written as
	\begin{align}
	\label{eq:S_optimization}
	& \underset{_{\boldsymbol{S}}}{\max} \quad \mbox{tr}\left(\left(\boldsymbol{S}\boldsymbol{\bar{P}}\boldsymbol{S}^{*}\right)^{-1}\boldsymbol{S}\boldsymbol{\bar{P}}\boldsymbol{\bar{W}}\boldsymbol{\bar{P}}\boldsymbol{S}^{*}\right)\\\label{eq:S Power Constraint} 
	& \mbox{s.t.}\qquad  \mbox{diag}^{\text{-}1}\left(\boldsymbol{S}^{*}\boldsymbol{S}\right)\leq\mathcal{P}\cdot\boldsymbol{1}_{M\cdot K}.
	\end{align}
	
	In the following theorem, we derive an optimal solution $\boldsymbol{S}_{\rm opt}$, which we refer to as the eigen-pilots. 
	\begin{thm}
		\label{thm:S_optimization}
		A pilot symbol matrix which solves \eqref{eq:S_optimization}-\eqref{eq:S Power Constraint} is given by
		\begin{align}
		\boldsymbol{S}_{\rm opt} & =\sqrt{\mathcal{P}}\boldsymbol{U}_{1}^{*},\label{eq:S}
		\end{align}
		with $\boldsymbol{U}_{1}$ the $M K \times \Tpilots$ matrix whose columns are the $\Tpilots$ eigenvectors of the Hermitian matrix  $\boldsymbol{\bar{W}}\boldsymbol{\bar{P}}$ corresponding to the $\Tpilots$ largest eigenvalues.
	\end{thm}
	\begin{proof}
		Let $\boldsymbol{Q}=\boldsymbol{S}\boldsymbol{\bar{P}}^{\frac{1}{2}}$.
		Then the objective in (\ref{eq:S_optimization}) becomes $\mbox{tr}\left(\boldsymbol{P_{\boldsymbol{Q}}}\boldsymbol{\boldsymbol{\bar{W}}\boldsymbol{\bar{P}}}\right)$,
		where $\boldsymbol{P_{\boldsymbol{Q}}}=\boldsymbol{Q}^{*}\left(\boldsymbol{Q}\boldsymbol{Q}^{*}\right)^{-1}\boldsymbol{Q}$
		is the orthogonal projection onto $\mathcal{R}\left(\boldsymbol{Q}^{*}\right)$, the range space of $\boldsymbol{Q}^{*}$.
		Let $\boldsymbol{V}\in\mathbb{C}^{MK\times\tau}$ be a matrix with orthogonal
		columns that spans $\mathcal{R}\left(\boldsymbol{Q}^{*}\right)$. Then
		$\boldsymbol{P_{\boldsymbol{Q}}}=\boldsymbol{V}\boldsymbol{V}^{*}$ and $\boldsymbol{Q}=\boldsymbol{D}\boldsymbol{V}^{*}$
		for some invertible $\tau\times\tau$ matrix $\boldsymbol{D}$. 
		
		If we ignore the constraint (\ref{eq:S Power Constraint}), then (\ref{eq:S_optimization})
		is solved by choosing $\boldsymbol{S}=\boldsymbol{D}\boldsymbol{V}^{*}\boldsymbol{\bar{P}}^{-\frac{1}{2}}$, where $\boldsymbol{V}$ is the solution of
		\begin{eqnarray}
		& \underset{{V}}{\max} & \mbox{tr}\left(\boldsymbol{V}^{*}\boldsymbol{\boldsymbol{\bar{W}}\boldsymbol{\bar{P}}}\boldsymbol{V}\right)\label{eq:Q_optimization}\\
		& \mbox{s.t.}\;\; & \boldsymbol{V}^{*}\boldsymbol{V}=\boldsymbol{I}_{\tau}.\nonumber
		\end{eqnarray}
		This problem has a closed form solution, $\boldsymbol{V}=\boldsymbol{U}_{1}$ 
		\cite[Ch. 20.A.2]{marshall_inequalities:_????}. 
		For this choice of $\boldsymbol{V}$,\textbf{ }the objective in (\ref{eq:Q_optimization})
		becomes 
		\begin{align}
		\mbox{tr}\left(\boldsymbol{V}^{*}\boldsymbol{\boldsymbol{\bar{W}}\boldsymbol{\bar{P}}}\boldsymbol{V}\right) & =\sum_{i=1}^{\tau}\lambda_{i}\left(\boldsymbol{\bar{W}}\boldsymbol{\bar{P}}\right).\label{eq:ObjectiveOptimalValue}
		\end{align}
		
		The expression (\ref{eq:ObjectiveOptimalValue}) is an upper bound on (\ref{eq:S_optimization})
		under the power constraint (\ref{eq:S Power Constraint}).
		We now show that there is a specific choice of $\boldsymbol{D}$ such that     $\boldsymbol{S}=\boldsymbol{D}\boldsymbol{V}^{*}\boldsymbol{\bar{P}}^{-\frac{1}{2}}$ both satisfies the power constraint and achieves the upper bound, and is therefore optimal. 
		First, we note that 
		\begin{equation}
		\boldsymbol{\bar{W}}\boldsymbol{\bar{P}}=\text{blkdiag}(w_1\boldsymbol{P}_1,\cdots,w_M\boldsymbol{P}_{M}).\label{eq:PW}
		\end{equation}
		Thus, $\boldsymbol{\bar{W}}\boldsymbol{\bar{P}}$ and $\boldsymbol{\bar{P}}$ share the same eigenvectors (possibly in different order).
		It then follows that $\boldsymbol{U}_{1}^{*}\boldsymbol{\bar{P}}^{-\frac{1}{2}}=\bar{\boldsymbol{\Sigma}}^{-\frac{1}{2}}\boldsymbol{U}_{1}^{*}$,
		where $\bar{\boldsymbol{\Sigma}}$ is a $\tau\times\tau$ diagonal matrix with the eigenvalues of $\boldsymbol{\bar{P}}$
		corresponding to $\boldsymbol{U}_{1}$ on its diagonal.
		Hence, $\boldsymbol{S}=\boldsymbol{D}\bar{\boldsymbol{\Sigma}}^{-\frac{1}{2}}\boldsymbol{U}_{1}^{*}$,
		and the power constraint (\ref{eq:S Power Constraint}) can be written
		in terms of $\boldsymbol{D}$ as 
		\begin{align}
		\mbox{diag}^{\text{-}1}\left(\boldsymbol{U}_{1}\bar{\boldsymbol{\Sigma}}^{-\frac{1}{2}}\boldsymbol{D^{*}}\boldsymbol{D}\bar{\boldsymbol{\Sigma}}^{-\frac{1}{2}}\boldsymbol{U}_{1}^{*}\right) & \leq\mathcal{P}\cdot\boldsymbol{1}_{M\cdot K}.\label{eq:D_constraint}
		\end{align}
		
		Let $\boldsymbol{\boldsymbol{D}}=\sqrt{\mathcal{P}}\bar{\boldsymbol{\Sigma}}^{\frac{1}{2}}$. 
		Then, $\mbox{diag}^{\text{-}1}\left(\boldsymbol{S}^{*}\boldsymbol{S}\right)=\mathcal{P}\mbox{diag}^{\text{-}1}\left(\boldsymbol{U}_{1}\boldsymbol{U}_{1}^{*}\right)$.
		Since $\boldsymbol{U}_{1}$ are columns of a unitary matrix,
		$\mbox{diag}\left(\boldsymbol{U}_{1}\boldsymbol{U}_{1}^{*}\right)$ is a vector
		with positive elements that are smaller or equal to 1, and the constraint (\ref{eq:S Power Constraint}) is satisfied. %
		It thus follows that $\boldsymbol{S}_{opt} =\sqrt{\mathcal{P}}\boldsymbol{U}_{1}^{*}$ is an optimal solution to (\ref{eq:S_optimization}), proving the theorem.    
	\end{proof}
	
	Note that the solution $\boldsymbol{S}_{opt}$ is not unique. 
	In fact, if $\boldsymbol{S}_{opt}$ is a solution to \eqref{eq:S_optimization}, then 
	\begin{equation}
	\boldsymbol{S}=\alpha\boldsymbol{T}\boldsymbol{S}_{op    t}
	\end{equation}
	is also a solution, where $\boldsymbol{T}$ is any $\Tpilots\times\Tpilots$ invertible matrix and $\alpha$ is chosen such that $\boldsymbol{S}$ complies with the power constraint. 
	For example, one can choose $\alpha=1$ and any unitary $\boldsymbol{T}$.
	
	Insight into the optimal pilot sequences in Theorem \ref{thm:S_optimization} can be obtained when considering the special case in which $\boldsymbol{P}_{i}$, $1\leq i\leq M$ are diagonal matrices.
	\begin{corollary}
		Let $\boldsymbol{P}_{i}$, $\Ind$ be diagonal matrices with $p_{i,kk}$ the $k$-th diagonal entry of $\boldsymbol{P}_{i}$. Then, the optimal solution $\boldsymbol{S}_{opt}$ corresponds to user selection, where the $\Tpilots$ \acp{ut} with strongest link $d_{ik}=w_{i}\cdot p_{i,kk}$ are chosen and assigned orthogonal sequences, while all other \acp{ut} are nullified.
	\end{corollary} 
	\begin{proof}
		In this case, $\boldsymbol{\bar{W}}\boldsymbol{\bar{P}}$ is a diagonal matrix with diagonal elements $d_{ik}=w_{i}\cdot p_{i,kk}$, $1\leq i\leq M$, $1\leq k\leq K$.
		Furthermore, the columns of $\boldsymbol{U}_{1}$ are unit vectors.
		Therefore, although our problem is not a user scheduling problem but pilot design, the resulting $\boldsymbol{S}_{\rm opt}$ corresponds to a user selection solution.
		For each \ac{ut} $k$ in cell $i$, we calculate its corresponding $d_{ik}$, and choose the $\Tpilots$ \acp{ut} with largest $d_{ik}$
		values, and the remaining $MK-\Tpilots$ \acp{ut} are nullified. 
	\end{proof}
	We emphasize that in contrast to classical user scheduling algorithms, e.g. \cite{gesbert_shifting_2007,shirani-mehr_joint_2011}, here the \acp{ut} are selected according to their long-term statistics regardless of the specific channel realization. 
	
	This specific case demonstrates some of the main differences between our design framework and the pilot allocation approach, e.g., \cite{yin_coordinated_2013,chen_pilot_2016,su_fractional_2015,yan_pilot_2015,zhu_smart_2015,fernandes_inter-cell_2013}. In allocation solutions, all the $MK$ \acp{ut} are necessarily transmitting, and the optimization is over how to choose the users that will be assigned the same pilot sequence. In contrast, in our solution, some \acp{ut} may be nullified. 
	The main disadvantage of this approach is its fairness across different \acp{ut}. 
	If one wishes to promote fairness in the system, alternative objectives than the sum in (\ref{eq:GeneralOptimization}), for example max-min, may be preferable \cite{kwan_proportional_2009,viswanath_opportunistic_2002,huh_multi-cell_2011,abhyankar_min-max_2007}. Nevertheless, if aiming at minimizing the overall \ac{mse} in the system, the eigen-pilots of (\ref{eq:S}) are optimal (for the fully-separable case).
	
	In the simulations in Section~\ref{sec:Esperiments}, we compare our eigen-pilots solution with other pilot design and allocation methods and show that it results in lower sum-\acp{mse} than other algorithms.
	
	
	\subsection{Partially Separable Correlations}
	\label{sec:PartiallySep}
	We now consider the scenario of partially-separable correlations, in which the channel correlation matrices satisfy \eqref{eq:PartiallSeparable}. 
	In this case, $\boldsymbol{P}_{ij}$ depends both on the transmitter index $j$ and the receiver index $i$, hence, instead of \eqref{eq:TransmitMatrix}, 
	we have the  $MK\times MK$  block diagonal matrix 
	\begin{equation}
	\boldsymbol{\bar{P}}_{i}=\mbox{blkdiag}\left(\boldsymbol{P}_{i1},\cdots,\boldsymbol{P}_{iM}\right).
	\end{equation}
	Let $\boldsymbol{Z}_i$ be an $M \times M$ matrix whose entries are zero except for the $i$-th diagonal entry, which equals one, and set $\boldsymbol{L}_i \triangleq \boldsymbol{Z}_i \otimes \boldsymbol{I}_{K}$.
	Then, $\boldsymbol{\bar{P}}_{i}\boldsymbol{L}_{i}$ is a block-diagonal matrix whose entries are zero except for the $i$-th block, which equals $\boldsymbol{P}_{i,i}$.
	Using these notations, and recalling the definition of $\boldsymbol{S}$  in \eqref{eq:PilotsMatrix}, we can rewrite \eqref{eq:GeneralOptimization3} as 
	\begin{eqnarray}
	\label{eq:SumOfRatiosPilots}
	&\underset{_{{\boldsymbol{S}}}}{\max} & \sum_{i}w_i \cdot \mbox{tr}\left(\boldsymbol{S}\boldsymbol{\bar{P}}_{i}^2\boldsymbol{L}_{i}\boldsymbol{S}^{*}\left[\boldsymbol{S}\boldsymbol{\bar{P}}_{i}\boldsymbol{S}^{*}\right]^{-1}\right)\\
	&\mbox{s.t.}\;\; & \mbox{diag}^{\text{-}1}\left(\boldsymbol{S}^{*}\boldsymbol{S}\right)\leq\mathcal{P}\cdot\boldsymbol{1}.\notag
	\end{eqnarray}
	
	Note that different from \eqref{eq:S_optimization}, here we still have the sum over the receiver index $i$. For $\boldsymbol{P}_{ij}=\boldsymbol{P}_{j}$ as in the fully-separable case, we get $\boldsymbol{\bar{P}}_{i}=\boldsymbol{\bar{P}}$, and \eqref{eq:SumOfRatiosPilots} coincides with \eqref{eq:S_optimization}.

	While, to the best of our knowledge, there is no known solution for \eqref{eq:SumOfRatiosPilots} in general, there are some special cases for which a solution is available. 
	For a single cell network, $M=1$, \eqref{eq:SumOfRatiosPilots} becomes the ratio-trace problem as in \eqref{eq:S_optimization}, and can be solved using the same method. 
	Another interesting scenario is when the number of pilot symbols is fixed to $\Tpilots = 1$. Although this setting is unlikely in massive \ac{mimo} systems, in which $\Tpilots$ is typically larger than $K$, we use this special case to illustrate the complexity of the considered problem. For a single pilot symbol the pilot symbols matrix $\boldsymbol{S}$ is an $1 \times MK$ vector $\boldsymbol{s}$, and \eqref{eq:SumOfRatiosPilots} becomes
	\begin{eqnarray}
	\label{eq:SumOfQuadRatios}
	&\underset{{\boldsymbol{s}}}{\max} & \sum_{i}w_i\frac{\boldsymbol{s}\boldsymbol{\bar{P}}_{i}^2\boldsymbol{L}_{i}\boldsymbol{s}^{*}}{\boldsymbol{s}\boldsymbol{\bar{P}}_{i}\boldsymbol{s}^{*}}\\
	&\mbox{s.t.} \;\;& s_{l}^{*}s_{l}\leq\mathcal{P},l=1\cdots MK.\notag
	\end{eqnarray}
	Problem \eqref{eq:SumOfQuadRatios} is known as the sum of quadratic ratios maximization, and has been addressed in previous works via branch and bound \cite{shen_solving_2009,shen_maximizing_2013} and harmony search \cite{jaberipour_solving_2010} algorithms. 
	However, the large dimension of $\boldsymbol{s}$ in (\ref{eq:SumOfQuadRatios}), that is equal to the total number of users in the system $MK$,  and the singularity in the origin makes these (and other similar) methods difficult to apply, even when $\Tpilots=1$. 
	
	In the following 
	we propose a greedy method, referred to as \ac{GSRTM}, for maximizing  \eqref{eq:SumOfRatiosPilots}, where we begin by considering pilot sequences of length $\tau=1$, and at each step we expand the sequences with one additional symbol. We choose the optimal symbol to add to each \ac{ut}'s sequence,
	by solving a vector problem similar to \eqref{eq:SumOfQuadRatios}. 
	The suggested \ac{GSRTM} method extends the \ac{GRTM} algorithm proposed in \cite{ioushua_eldar_hybrid_2017} to the case of sum of ratios trace as in \eqref{eq:SumOfRatiosPilots}. 
	
	In the simulation section, we compare the performance of GSRTM with other pilot design and allocation methods, and demonstrate that it results in lower sum-\acp{mse} than the other methods in the partially-correlated case.
	
	\subsection{GSRTM - Greedy Sum of Ratio Traces Maximization} 
	In order to formulate the \ac{GSRTM} algorithm, we first note that due to the quadratic form of the objective in \eqref{eq:SumOfRatiosPilots}, it is not sensitive to a scale in $\boldsymbol{S}$. That is, the objective value corresponding to a given $\boldsymbol{S}$ is identical to the value that corresponds to $\alpha\boldsymbol{S}$, for any scalar $\alpha\neq 0$.
	Therefore, we can first solve \eqref{eq:SumOfRatiosPilots} without the power constraint, and then scale the solution to comply with it.
	
	Assume we have a solution $\Sk{N}\in\mathbb{C}^{N\times MK}$ for the $N$-sized problem, i.e. (\ref{eq:SumOfRatiosPilots}) with $\tau=N$. 
	We now wish to add an additional pilot symbol, 
	and compute the optimal row vector $\boldsymbol{s}\in\mathbb{C}^{1\times MK}$ to add to the previously selected sequences such that $\Sk{N+1}=[\Sk{N}^T\;\boldsymbol{s}^{T}]^{T}\in\mathbb{C}^{N+1\times MK}$. 
	
	First, note that in order for the objective in (\ref{eq:SumOfRatiosPilots}) to be well defined in the $\left(N+1\right)$-sized case, we require $\Sk{N+1}\boldsymbol{\bar{P}}_{i}\Sk{N+1}^*$ to be invertible for all $1\leq i\leq M$, namely, that $\boldsymbol{s}\notin\mathcal{R}(\Sk{N})$. 
	In practice, this condition implies that the new additional pilot symbols vector $\boldsymbol{s}$ contributes to the orthogonality between pilot sequences of different users.
	
	Given this condition, the $\left(N+1\right)$-sized optimization problem is given by
	\begin{equation}
	\begin{aligned}
	&\underset{s}{\max} & \sum_{i}w_i\mbox{tr}\Big(\Sk{N+1}&\boldsymbol{A}_{i}\Sk{N+1}^*\left[\Sk{N+1}\boldsymbol{B}_{i}\Sk{N+1}^*\right]^{-1}\Big) \label{eq:KOpt}\\
	&\mbox{s.t.} &\boldsymbol{s}\notin\mathcal{R}(\Sk{N}),\;\;\;\;\;& 
	\end{aligned}
	\end{equation}
	with $\boldsymbol{A}_i\triangleq\boldsymbol{\bar{P}}_{i}^2\boldsymbol{L}_{i}$ and $\boldsymbol{B}_i\triangleq\boldsymbol{\bar{P}}_{i}$.
	
	In the next proposition we show that for a specific choice of  $MK\times MK$ \ac{psd} matrices  $\{\boldsymbol{G}_{i}\}_{i=1}^M$ and $\{\boldsymbol{T}_{i}\}_{i=1}^M$, solving \eqref{eq:KOpt} is equivalent to solving the following sum of quadratic ratios problem: 
	\begin{eqnarray}
	\label{eq:VectorProblem}
	&\underset{\boldsymbol{s}}{\max} & \sum_{i}w_i\frac{\boldsymbol{s}\boldsymbol{G}_{i}\boldsymbol{s}^{*}}{\boldsymbol{s}\boldsymbol{T}_{i}\boldsymbol{s}^{*}}\\
	&\mbox{s.t.}\;\; & \boldsymbol{s}\boldsymbol{T}_{i}\boldsymbol{s}^{*}>0,\quad 1\leq i\leq M.
	\notag
	\end{eqnarray}
	Proposition \ref{prop:GRTM} leads to the formulation of the GSRTM algorithm for approaching the solution of \eqref{eq:SumOfRatiosPilots}, where in each iteration we choose the best row vector to add to the previously selected pilot matrix rows by solving \eqref{eq:VectorProblem}. 
	Once all $\Tpilots$ rows of $\boldsymbol{S}$ are chosen, the entire matrix is scaled to comply with the power constraint.
	The GSRTM algorithm is summarized in Algorithm \ref{algo:Greedy}.
	
	\begin{prop}\label{prop:GRTM}
		Problems (\ref{eq:VectorProblem}) and  (\ref{eq:KOpt}) are equivalent with
		\begin{equation*}
		\begin{aligned}
		&    \boldsymbol{T}_{i}=\boldsymbol{B}_{i}^{\frac{1}{2}}\left(\boldsymbol{I}_{MK}-\boldsymbol{P}_{\boldsymbol{B}_{i}^{\frac{1}{2}}\Sk{k}^*}\right)\boldsymbol{B}_{i}^{\frac{1}{2}}\\
		&\gamma_{i}=\text{tr}\left(\boldsymbol{P}_{\boldsymbol{B}_{i}^{\frac{1}{2}}\Sk{k}^*}\boldsymbol{B}_{i}^{-\frac{1}{2}}\boldsymbol{A}_{i}^{\frac{1}{2}}\boldsymbol{B}_{i}^{-\frac{1}{2}}\right)\\ 
		&\boldsymbol{X}_{i}=\boldsymbol{B}_{i}^{\frac{1}{2}}\boldsymbol{P}_{\boldsymbol{B}_{i}^{\frac{1}{2}}\Sk{k}^*} \boldsymbol{B}_{i}^{-\frac{1}{2}}\boldsymbol{A}_{i}^{\frac{1}{2}}-\boldsymbol{A}_{i}^{\frac{1}{2}}\\
		&\boldsymbol{G}_{i}=\gamma_{i}\boldsymbol{T}_{i}+\boldsymbol{X}_{i}\boldsymbol{X}_{i}^{*} .
		\end{aligned}
		\end{equation*}
	\end{prop}
	We note that the proof of Proposition~\ref{prop:GRTM} detailed in the following is similar to the proof of \cite[Prop. 3]{ioushua_eldar_hybrid_2017}. However, here the proof takes into account a sum of ratio traces, rather than a single ratio trace, and considers a different constraint on $\boldsymbol{s}$.
	\begin{proof}
		To prove the proposition, we rely on the following lemma: 
		\begin{lem}\cite[Ch. 3]{Petersen06thematrix}\label{thm:rank1updtae}
			Let $\tilde{\boldsymbol{S}}=\left[\boldsymbol{S}^T\;\boldsymbol{s}^T\right]^T$ and denote $\boldsymbol{Q}=\left(\boldsymbol{S}\boldsymbol{S}^{*}\right)^{-1}$. Then
			\begin{equation*}
			\left(\tilde{\boldsymbol{S}}\tilde{\boldsymbol{S}^{*}}\right)^{-1}=\left[
			\begin{array}{cc}
			\boldsymbol{Q}+\alpha\boldsymbol{Q}\boldsymbol{S}\boldsymbol{s}^{*}\boldsymbol{s}\boldsymbol{S}^{*}\boldsymbol{Q}^{*} & -\alpha\boldsymbol{Q}\boldsymbol{S}\boldsymbol{s}^{*}\\
			-\alpha\boldsymbol{s}\boldsymbol{S}^{*}\boldsymbol{Q}^{*} & \alpha
			\end{array}
			\right],
			\end{equation*}
			with $\alpha=\frac{1}{\boldsymbol{s}\boldsymbol{s}^{*}-\boldsymbol{s}\boldsymbol{S}^{*}\boldsymbol{Q}\boldsymbol{S}\boldsymbol{s}^{*}}$.
		\end{lem}
		Using Lemma~\ref{thm:rank1updtae} and straightforward algebraic operations we get equality between the $i$-th summand of the sums in \eqref{eq:KOpt} and in \eqref{eq:VectorProblem}. Consequently,  these two objectives are equal, and it remains to show that the constraints are equivalent to prove the proposition.
		
		Since $\boldsymbol{B}_{i}$, $1\leq i\leq M$, are invertible, $\boldsymbol{s}^{*}\in\mathcal{R}(\Sk{N}^*)$ if and only if $\boldsymbol{B}_{i}^{\frac{1}{2}}\boldsymbol{s}^{*}\in\mathcal{R}(\boldsymbol{B}_{i}^{\frac{1}{2}}\Sk{N}^* )$, $1\leq i\leq M$. 
		Therefore, if $\boldsymbol{s}^{*}\in\mathcal{R}(\Sk{N}^*)$, then
		\begin{align*}
		\boldsymbol{s}^{*}\boldsymbol{T}_{i}\boldsymbol{s}^{*}
		&=\boldsymbol{s}\boldsymbol{B}_{i}\boldsymbol{s}-\boldsymbol{s}^{*}\boldsymbol{B}_{i}^{\frac{1}{2}} \boldsymbol{P}_{\boldsymbol{B}_{i}^{\frac{1}{2}}\Sk{N}^*}\boldsymbol{B}_{i}^{\frac{1}{2}}\boldsymbol{s}\\
		&=\boldsymbol{s}^{*}\boldsymbol{B}_{i}\boldsymbol{s}^{*}-\boldsymbol{s}\boldsymbol{B}_{i}^{\frac{1}{2}} \boldsymbol{B}_{i}^{\frac{1}{2}}\boldsymbol{s}^{*}=0,\quad \forall    1\leq i\leq M.    
		\end{align*}
		In the other direction, note that $\boldsymbol{T}_{i}$ is non-negative Hermitian matrix, and hence can be decomposed as $\boldsymbol{T}_{i}=\boldsymbol{Q}_{i}\boldsymbol{Q}_{i}^{*}$ for some $\boldsymbol{Q}_{i}$. Thus, $\boldsymbol{s}\boldsymbol{T}_{i}\boldsymbol{s}=0$ if and only if $\boldsymbol{Q}_{i}^{*}\boldsymbol{s}=\boldsymbol{0}$. Multiplying both sides of the equation by $\boldsymbol{Q}_{i}$ yields
		\begin{equation*}
		\boldsymbol{T}_{i}\boldsymbol{s}^{*}=\boldsymbol{B}_{i}^{\frac{1}{2}}\Big(\boldsymbol{I}_{MK}-\boldsymbol{P}_{\boldsymbol{B}_{i}^{\frac{1}{2}}\Sk{N}^*}\Big)\boldsymbol{B}_{i}^{\frac{1}{2}}\boldsymbol{s}^{*}=\boldsymbol{0}.
		\end{equation*}
		Since $\boldsymbol{B}_{i}$ is invertible, it follows that \\
		$(\boldsymbol{I}_{MK}-\boldsymbol{P}_{\boldsymbol{B}_{i}^{\frac{1}{2}}\Sk{N}^*})\boldsymbol{B}_{i}^{\frac{1}{2}}\boldsymbol{s}^{*}=\boldsymbol{0}$,
		or, $\boldsymbol{B}_{i}^{\frac{1}{2}}\boldsymbol{s}^{*}\in\mathcal{R}(\boldsymbol{B}_{i}^{\frac{1}{2}}\Sk{N}^*)$, which is equivalent to $\boldsymbol{s}^{*}\in\mathcal{R}\big(\Sk{N}^*\big)$.  
		Hence, we proved that
		\begin{equation}
		\label{eqn:tempApp1}
		\boldsymbol{s}^{*}\in\mathcal{R}\left(\Sk{K}^*\right)\iff\boldsymbol{s}\boldsymbol{T}_{i}\boldsymbol{s}^{*}=0,\;\forall 1\leq i\leq M.
		\end{equation}
		Since $\boldsymbol{T}_{i}$ is non-negative definite, $\boldsymbol{s}\boldsymbol{T}_{i}\boldsymbol{s}^{*}\geq0$, for all $\boldsymbol{s}$. 
		The previous connection then yields
		\begin{equation*}
		\boldsymbol{s}^{*}\notin\mathcal{R}\left(\Sk{K}^*\right)\iff\boldsymbol{s}\boldsymbol{T}_{i}\boldsymbol{s}^{*}>0,\;\forall 1\leq i\leq M.
		\end{equation*}     
		As there is equality between the objectives and feasible sets of (\ref{eq:VectorProblem}) and  (\ref{eq:KOpt}), both problems are equivalent.  
	\end{proof}
	\smallskip
	Our remaining task is now to solve the sum of quadratic rations problem  in \eqref{eq:VectorProblem}. As mentioned before, a possible approach can be to use one of the algorithms \cite{shen_solving_2009,shen_maximizing_2013,jaberipour_solving_2010}, but due to the large dimension of $\boldsymbol{s}$, these will result in very large runtime.
	
	One simple method for obtaining a (sub-optimal) solution for \eqref{eq:VectorProblem} is searching over a selected dictionary, and calculating the objective value in \eqref{eq:VectorProblem} for each vector in the dictionary. 
	The vector that corresponds to the largest objective value is then chosen. 
	The computational load of this approach is reduced both compared to \cite{shen_solving_2009,shen_maximizing_2013,jaberipour_solving_2010}, and to applying similar dictionary solutions directly on (\ref{eq:KOpt}), since here each element in the sum is a scalar ratio rather than a matrix ratio. 
	The choice of dictionary is flexible. In the simulations study detailed in Section \ref{sec:Esperiments}, we compare different possible such choices. 
	
	\begin{algorithm} 
		\caption{GSRTM}
		\text{\textbf{Input} correlation matrices $\bar{\boldsymbol{P}}_{i}$, $1\leq i\leq M$}\\
		\text{\textbf{Output} pilot matrix $\boldsymbol{S}$}\\
		\text{\textbf{Initialize} $\boldsymbol{A}_i\;\;=\bar{\boldsymbol{P}}_{i}^{2}\boldsymbol{L}_{i}$,\; $\boldsymbol{B}_{i}=\bar{\boldsymbol{P}}_{i}$,\; $1\leq i\leq M$} \\
		\text{\;\;\;\;\;\;\;\;\;\;\;\;\;\;
			$\boldsymbol{G}_{i}\;\;=\boldsymbol{A}_{i},\;\;\;\;\;\;\boldsymbol{T}_{i}=\boldsymbol{B}_i$, \;$1\leq i\leq M$}
		\\ \text{$\;\;\;\;\;\;\;\;\;\;\;\;\;\;\;\Sk{0}=[]$}\\
		\text{\textbf{For} $K=0:\tau-1$:}
		\begin{enumerate}
			\item \text{\textbf{Solve} the base case (\ref{eq:VectorProblem}) to obtain $\boldsymbol{s}$}
			\item \text{\textbf{Update} $\Sk{K+1}=\left[\Sk{K}^T\;\boldsymbol{s}^{T}\right]^{T}$}
			\item \text{\textbf{Update} }
			\begin{equation*}
			\begin{aligned}
			&    \boldsymbol{T}_{i}=\boldsymbol{B}_{i}^{\frac{1}{2}}\left(\boldsymbol{I}_{MK}-\boldsymbol{P}_{\boldsymbol{B}_{i}^{\frac{1}{2}}\Sk{K}^*}\right)\boldsymbol{B}_{i}^{\frac{1}{2}}\\
			&\gamma_{i}=\text{tr}\left(\boldsymbol{P}_{\boldsymbol{B}_{i}^{\frac{1}{2}}\Sk{K}^*}\boldsymbol{B}_{i}^{-\frac{1}{2}}\boldsymbol{A}_{i}^{\frac{1}{2}}\boldsymbol{B}_{i}^{-\frac{1}{2}}\right)\\ 
			&\boldsymbol{X}_{i}=\boldsymbol{B}_{i}^{\frac{1}{2}}\boldsymbol{P}_{\boldsymbol{B}_{i}^{\frac{1}{2}}\Sk{K}^*} \boldsymbol{B}_{i}^{-\frac{1}{2}}\boldsymbol{A}_{i}^{\frac{1}{2}}-\boldsymbol{A}_{i}^{\frac{1}{2}}\\
			&\boldsymbol{G}_{i}=\gamma_{i}\boldsymbol{T}_{i}+\boldsymbol{X}_{i}\boldsymbol{X}_{i}^{*} \\
			\end{aligned}
			\end{equation*}
		\end{enumerate}\label{algo:Greedy}
	\end{algorithm}

	\section{Numerical Experiments }
	\label{sec:Esperiments}
	We now present some numerical results that both demonstrate the benefit of our joint pilot design framework and the effect of the RF reduction on the system's performance. 
	First,  in Section~\ref{subsec:fullSep} we consider the fully-separable correlations case with different combiner methods and show that the eigen-pilots of Section~\ref{sec:FullySep} enjoys lower MSE than other pilot allocation and design techniques. In Section \ref{subsec:partSep} we consider partially-separable correlations, compare the GSRTM of Algorithm~\ref{algo:Greedy} for pilot sequence design with other methods and show that in this case as well the joint design framework yields lower MSE.
	
	Since the combiner design problem is identical for both correlation scenarios, the effect of RF reduction will be tested only under the fully-separable case. In the partially-separable experiments a full combiner, i.e. $\Nrf=\Nbs$ is considered. 
	
	We compare our proposed algorithms with three other approaches: 
	the common reused orthogonal pilots, where at all the cells the same $K$ orthogonal pilots are transmitted without any optimization over the allocation order, 
	the random pilot design, where at each cell $i$ the pilot matrix $\boldsymbol{S}_{i}$ is drawn randomly
	from an i.i.d. zero-mean unit variance complex-normal distribution, and the smart pilot assignment (SPA) method from \cite{zhu_smart_2015}, which is a near-optimal allocation algorithm when $\tau=K$. 
	As mentioned before, pilot design can also be performed from a single cell perspective as in \cite{pang_optimal_2007,liu_training_2007,kotecha_transmit_2004,noh_pilot_2014,bogale_pilot_2014,bjornson_framework_2010}. However we do not compare our method with this approach, as when generalizing this framework to the multi-cell setting by iterative design, the solution rarely converges. 
	
	Notice that in the reused orthogonal pilots, only $K$ orthogonal sequences are used across all cells, regardless of the number of pilot symbols  $\Tpilots$.
	In practice, when $\Tpilots>K$, there are more than $K$ possible orthogonal sequences, and the reuse ratio can be decreased. 
	This scenario is similar to the random pilot case, where with high probability $\Tpilots$ orthogonal pilots will be drawn. 
	In all techniques, the pilots are normalized to comply with the power constraint $\mathcal{P}=1$. 
	
	Throughout this section, the performance measure used is the sum of normalized \acp{mse}, defined as $\bar{\epsilon}=\frac{1}{M}\sum_{i=1}^{M}\frac{\Vert \boldsymbol{h}_{i}-\hat{\boldsymbol{h}}_{i}\Vert^{2}}{\Vert \boldsymbol{h}_{i}\Vert^{2}}$.
	
	\subsection{Fully Separable Model Experiments}
	\label{subsec:fullSep}
	We begin by treating the fully-separable channel correlation model \eqref{eq:FullySeparable} and test the performance of the eigen-pilot method of Section~\ref{sec:FullySep}. 
	
	We use $M=3$ cells, $\Nbs=10$ \ac{bs} antennas, $K=4$ \acp{ut} in each cells and a single digital input at the \acp{bs}, i.e., $\Nrf=1$. 
	We test two analog combiners: A fully-digital combiner \eqref{eq:WfullyDigital} with $\boldsymbol{T}=\boldsymbol{I}_{\Nrf}$, and an analog combiner obtained using the \ac{GRTM} algorithm proposed in  \cite{ioushua_eldar_hybrid_2017}. \ac{GRTM} is a greedy dictionary-based method, that at each step adds one RF chain to the system and chooses the optimal combiner vector to add to the previously selected ones from a feasible dictionary. 
	We assume a fully-connected phase shifter network, namely, the feasible set is all matrices with unimodular entries.
	The choice of combiner defines the weights $w_i$ in \eqref{eq:wights} that will be used for the pilot design.
	
	The first experiment compares between the eigen-pilots, orthogonal pilot reuse, and random pilots methods, with fully-digital and GRTM based analog combiners. 
	We set the receive side correlation $\boldsymbol{Q}_{i}$ to be a random matrix given by $\boldsymbol{Q}_{i}=\boldsymbol{X}_{i}\boldsymbol{X}_{i}^{*}$, where $\{\boldsymbol{X}_{i}\}_{i=1}^M$ is a set of independent random matrices with i.i.d zero-mean unit variance complex-Gaussian entries. 
	The transmit side correlation $\boldsymbol{P}_{j}$ is a diagonal matrix with diagonal elements drawn uniformly from the interval $\left[0,1\right]$.
	A new realization of $\boldsymbol{Q}_{i},\boldsymbol{P}_{j}$ is generated for each of the 10000 Monte Carlo simulations.
	
	Figure~\ref{fig:KroneckerCombiners} depicts the sum of normalized estimation error $\bar{\epsilon}$ for different number of pilot symbols. 
	While additional symbols improve the eigen-pilot performance, as it enables choosing more users, it does not affect the reused pilot performance. 
	This is because when there is full pilot reuse between cells, additional symbols can only improve intra-cell interference, which in this case does not occur. 
	The random pilot design improves with additional symbols, but falls short in comparison to the eigen-pilots, 
	up to the point when full orthogonality is achieved in both techniques, i.e. $\Tpilots=MK = 12$. 
	As expected, for all methods the fully-digital combiner experiences better performance than the analog GRTM one.
	
	\begin{figure}
		\centering{}\includegraphics[width=0.39\textwidth]{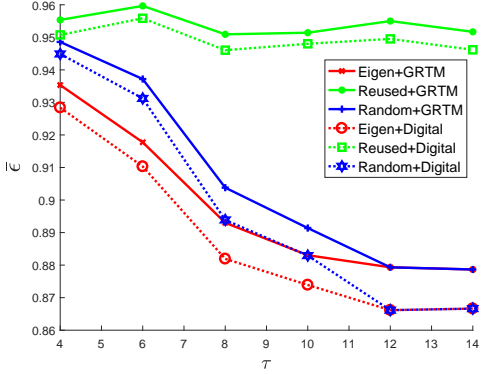} \vspace{-0.4cm}
		\vspace{0.1cm}
		\protect\protect\caption{{\footnotesize{}Estimation error of different pilot design methods vs. number of pilot symbols $\tau$, with GRTM and fully-digital combiners and general receive correlation structure.
			}\label{fig:KroneckerCombiners}}
		\vspace{-0.3cm}
	\end{figure}
	
	Next, we compare the eigen-pilots method with SPA, which is a near-optimal pilot assignment method when $K=\tau$ \cite{zhu_smart_2015}.	
	The SPA framework was developed assuming the MU-MIMO channel fading model \eqref{eq:MarzetaModel}. 
	To comply with this model, we set
	\begin{equation*}
	\boldsymbol{P}_{j}=\text{diag}\left(\beta_{1j},\cdots,\beta_{Mj}\right),\;\boldsymbol{Q}_{i}=\boldsymbol{I}_{\Nbs}
	\end{equation*}
	where $\beta_{i,j}$ are drawn uniformly from the interval $\left[0,1\right]$.
	For this choice, the fully-separable model \eqref{eq:FullySeparable} coincides with the correlation model \eqref{eq:MarzetaModel}, where $\boldsymbol{D}_{ij}=\boldsymbol{P}_{j}$.  
	The pilot sequences used for SPA are the first $\tau$ eigenvectors of $\boldsymbol{X}\boldsymbol{X}^{*}$, with $\boldsymbol{X}$ a random matrix with i.i.d complex-Normal entries with zero-mean and unit variance, that is drawn once every Monte-Carlo simulation.
	
	The resulting $\bar{\epsilon}$  versus the number of pilot symbols $\Tpilots$ is depicted in 
	Figure~\ref{fig:KroneckerSPA}.   
	It is observed in Figure~\ref{fig:KroneckerSPA} that when $K=\Tpilots=4$, SPA indeed outperforms the random pilots but falls short in comparison to the eigen-pilots.
	This result is expected as SPA optimizes only the allocation of the pilots and not the sequences themselves. While the comparison may be unfair, it demonstrates the fact that pilot allocation is a sub-optimal approach in general.
	To the best of our knowledge, no other works considered joint pilot design across all cells, hence more appropriate comparisons are not applicable. 
	
	Figure~\ref{fig:KroneckerSPA_rf} depicts the eigen-pilots, SPA and random-pilots performance versus different numbers of RF chains $\Nrf$, for a fully-digital combiner \eqref{eq:WfullyDigital} and a full-receiver without RF reduction, i.e. $\Nrf=\Nbs$.
	Here we set $\Tpilots=5$.
	Since  $\boldsymbol{R}_{r_i}=\boldsymbol{I}_{\Nbs}$, it follows from \eqref{eq:WfullyDigital} that one possible optimal RF reduction solution is antenna selection, where all antennas are equally important, and every additional RF chain enables choosing one more antenna and improves the estimation error linearly.
	Figure~\ref{fig:KroneckerSPA_rf} demonstrate the effect of the reduction on the different methods. It can be seen that the performance gap between the methods is growing larger as more RF chains are added to the system. For all $\Nrf$ values the eigen-pilots outperform the other methods.

	\begin{figure}[t]
		\centering{}\includegraphics[width=0.39\textwidth]{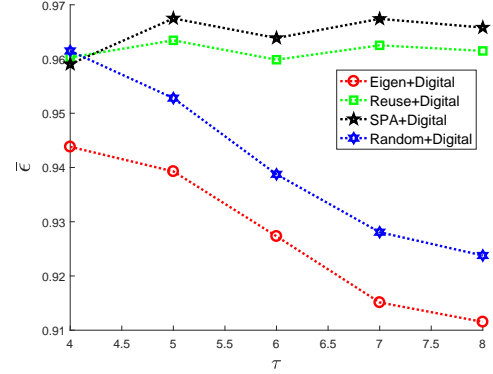} \vspace{-0.5cm}
		\vspace{0.1cm}
		\protect\protect\caption{{\footnotesize{}Estimation error of different pilot design methods vs.  number of pilot symbols $\tau$, with fully-digital combiner.}
			\label{fig:KroneckerSPA}}
		\vspace{-0.1cm}
	\end{figure}
	\vspace{0.2cm}
	\begin{figure}[t]
		\centering{}\includegraphics[width=0.39\textwidth]{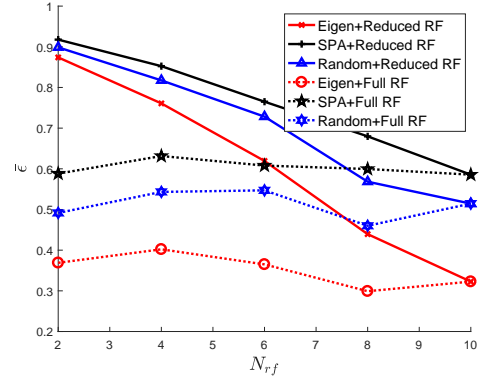} \vspace{-0.5cm}
		\vspace{0.2cm}
		\protect\protect\caption{{\footnotesize{}Estimation error of different pilot design methods vs. number of RF chains $\Nrf$, with fully-digital and full-RF combiners.}
			\label{fig:KroneckerSPA_rf}}
		\vspace{-0.1cm}
	\end{figure}

	\subsection{Partially Separable Model Experiments}
	\label{subsec:partSep}
	For the partially-separable correlations case, we adopt a the MU-MIMO channel fading model \cite{marzetta_noncooperative_2010}, such that $\beta_{ikj}$, $1\leq i,j\leq M$, $1\leq k\leq K$, is defined via 
	\begin{equation}
	\beta_{ikj}=\frac{z_{ikj}}{r_{ikj}^\gamma}
	\end{equation} 
	with $r_{ikj}$ representing the distance between the $k$-th \ac{ut} in the $j$-th cell and the $i$-th \ac{bs}, 
	$\gamma$ is the decay exponent, and $z_{ikj}$ is a log-normal random variable, 
	such that $10\log{z_{ikj}}$ is zero-mean Gaussian with a standard deviation $\sigma_{shad}$.
	Here we used $\gamma=3$ and $\sigma_{shad}=8 [dB]$. 
	The diagonal transmit correlation matrices are given  by $\boldsymbol{P}_{i,j}=\text{diag}\left(\beta_{i1j},\cdots,\beta_{iMj}\right)$.
	The MU-MIMO network consists of $M=7$ hexagonal cells in which the \acp{ut} are distributed uniformly. A realization of such a network is illustrated in Figure~\ref{fig:Cells}.
	\begin{figure}[t]
		\centering{}\includegraphics[width=0.39\textwidth]{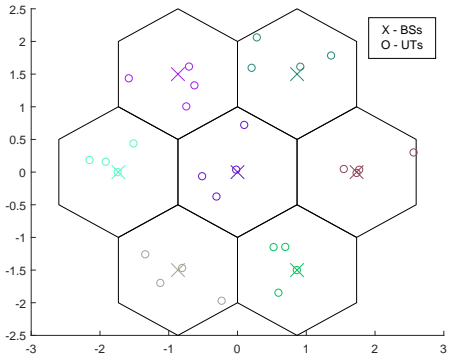} \vspace{-0.5cm}
		\protect\protect\caption{{\footnotesize{}MU-MIMO system with 7 hexagonal cells.}
			\label{fig:Cells}}
		\vspace{-0.5cm}
	\end{figure}
	
	We compare the proposed GSRTM algorithm for pilot sequence design to random pilot design and SPA. Here we drop the orthogonal pilot reuse method, as it is inferior to SPA in the sum-\acp{mse} sense.
	We use  $\Nbs=10$ antennas at each \ac{bs} and $K=4$ \acp{ut} in each cell.
	Since, as noted in Section~\ref{sec:PartiallySep}, the analog combiner design problem in the partially-separable correlation profile is identical to the fully-separable correlations scenario, in the following we focus only on the pilot sequence design algorithms, and consider a system without RF chain reduction, i.e., $\Nrf = \Nbs = 10$.
	
	To numerically evaluate the effect of the number of pilot symbols $\Tpilots$ on the performance of the considered methods, we depict in Figure~\ref{fig:PartiallyMethod} the estimation performance for $\Tpilots \in [4,8]$. Observing  Figure~\ref{fig:PartiallyMethod}, we note that while for $K=\Tpilots$ SPA achieves the best estimation accuracy, as the number of pilot symbols increases, GSRTM results in lower sum-\acp{mse}. This is since SPA cannot exploit additional symbols and is limited to the allocation of $K$ sequences alone.
	
	Next, we study the performance of GSRTM with different dictionaries. 
	Here we used three options: general Gaussian dictionary, where the dictionary matrix has i.i.d complex-Normal entries with zero-mean and unit variance, quadratic amplitude modulation (QAM) with 4 constellation points and QAM with 16 constellation points. All dictionaries are composed of $Q=300$ different possible sequences. That is, the difference between the dictionaries is not the number of possible sequences, but the flexibility of symbol values.
	As expected, the complex-Normal dictionary results in the lowest MSE, as it offers more flexibility and diversity.
	The QAM16 dictionary performance is not much lower than the Gaussian one, as its flexibility is quite high due to a large number of constellation points. Yet, it is much simpler to implement in practice than the infeasible Gaussian dictionary.
	As expected, QAM4 has the highest \ac{mse} as it is the most limited. However, it requires fewer representation bits for each symbol. 
	
	\begin{figure}[t]
		\centering{}\includegraphics[width=0.39\textwidth]{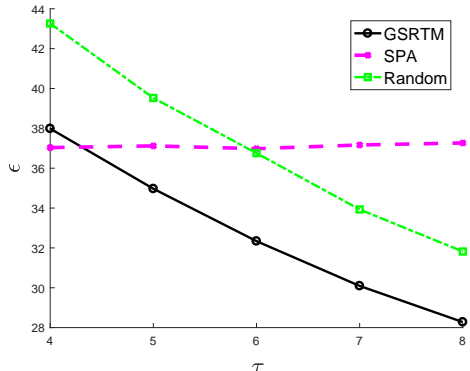} \vspace{-0.5cm}
		\vspace{0.1cm}
		\protect\protect\caption{{\footnotesize{}Estimation error of different pilot design methods vs. number of pilot symbols $\tau$.}
			\label{fig:PartiallyMethod}}
		\vspace{-0.5cm}
	\end{figure}
	\begin{figure}[t]
		\centering{}\includegraphics[width=0.39\textwidth]{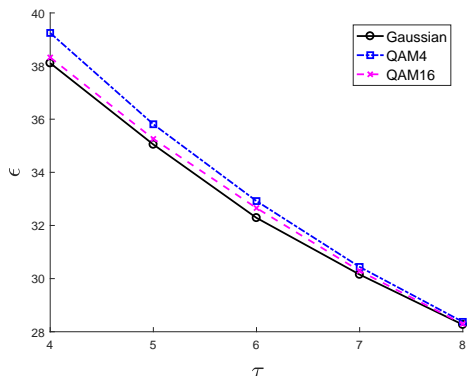} \vspace{-0.5cm}
		\vspace{0.1cm}
		\protect\protect\caption{{\footnotesize{}Estimation error of GSRTM with different dictionaries vs. number of pilot symbols $\tau$.}
			\label{fig:PartiallyDictionary}}
		\vspace{-0.5cm}
	\end{figure}
	
	\section{Conclusions}
	\label{sec:Conclusions}
	We treated the problem of joint pilot sequence and analog combiner design in massive MIMO systems with a reduced number of RF chains. 
	We considered two channel correlation profiles: fully- and partially-separable correlations.
	We showed that for the considered setup, the RF reduction problem can be solved separately from the pilot sequences design and that previously suggested reduction methods fit this framework as well.
	For the fully-separable case, we derived an optimal pilot design solution. 
	We demonstrated that in the special case when the transmit correlation matrices are diagonal, the solution corresponds to a user selection method, where the users are chosen according to their long-term statistics.
	For the partially-separable correlations scenario,  we suggested a greedy method for pilot design, which,  at each step, allows for one more pilot symbol to be added, and solves a sum-of-quadratic-ratios problem.
	Our numerical examples demonstrated the substantial benefits of our proposed designs compared to existing methods that do not jointly optimize the pilots.
	
	\section*{Acknowledgement}
	The authors would like to thank Dr. Nir Shlezinger for his valuable input.
	
	\bibliographystyle{ieeetr}
	\bibliography{PilotPaper}

\begin{thebibliography}{10}

\bibitem{andrews_what_2014}
J.~G. Andrews, S.~Buzzi, W.~Choi, S.~V. Hanly, A.~Lozano, A.~C.~K. Soong, and
  J.~C. Zhang, ``What {will} 5g {be}?,'' {\em IEEE Journal on Selected Areas in
  Communications}, vol.~32, pp.~1065--1082, June 2014.

\bibitem{rusek_scaling_2013}
F.~Rusek, D.~Persson, B.~K. Lau, E.~G. Larsson, T.~L. Marzetta, and
  F.~Tufvesson, ``Scaling {up} {MIMO}: {opportunities} and {challenges} with
  {very} {large} {arrays},'' {\em IEEE Signal Processing Magazine}, vol.~30,
  pp.~40--60, Jan. 2013.

\bibitem{marzetta_noncooperative_2010}
T.~L. Marzetta, ``Noncooperative {cellular} {wireless} with {unlimited}
  {numbers} of {base} {station} {antennas},'' {\em IEEE Transactions on
  Wireless Communications}, vol.~9, pp.~3590--3600, Nov. 2010.

\bibitem{hoydis_massive_2013}
J.~Hoydis, S.~ten Brink, and M.~Debbah, ``Massive {MIMO} in the {UL}/{DL} of
  {cellular} {networks}: {how} {many} {antennas} {do} {we} {need}?,'' {\em IEEE
  Journal on Selected Areas in Communications}, vol.~31, pp.~160--171, Feb.
  2013.

\bibitem{jose_pilot_2011}
J.~Jose, A.~Ashikhmin, T.~L. Marzetta, and S.~Vishwanath, ``Pilot
  {contamination} and {precoding} in {multi}-{cell} {TDD} {systems},'' {\em
  IEEE Transactions on Wireless Communications}, vol.~10, pp.~2640--2651, Aug.
  2011.

\bibitem{ashikhmin_pilot_2012}
A.~Ashikhmin and T.~Marzetta, ``Pilot contamination precoding in multi-cell
  large scale antenna systems,'' in {\em Information {Theory} {Proceedings}
  ({ISIT}), 2012 {IEEE} {International} {Symposium} on}, pp.~1137--1141, IEEE,
  2012.

\bibitem{ngo_evd-based_2012}
H.~Q. Ngo and E.~G. Larsson, ``{EVD}-based channel estimation in multicell
  multiuser {MIMO} systems with very large antenna arrays,'' in {\em {IEEE}
  {International} {Conference} on Acoustics, {Speech} and {Signal} {Processing}
  ({ICASSP})}, pp.~3249--3252, March 2012.

\bibitem{muller_blind_2014}
R.~R. Muller, L.~Cottatellucci, and M.~Vehkapera, ``Blind {pilot}
  {decontamination},'' {\em IEEE Journal of Selected Topics in Signal
  Processing}, vol.~8, pp.~773--786, Oct. 2014.

\bibitem{yin_coordinated_2013}
H.~Yin, D.~Gesbert, M.~Filippou, and Y.~Liu, ``A {coordinated} {approach} to
  {channel} {estimation} in {large}-{scale} {multiple}-{antenna} {systems},''
  {\em IEEE Journal on Selected Areas in Communications}, vol.~31,
  pp.~264--273, Feb. 2013.

\bibitem{chen_pilot_2016}
Z.~Chen and C.~Yang, ``Pilot {decontamination} in {wideband} {massive} {MIMO}
  {systems} by {exploiting} {channel} {sparsity},'' {\em IEEE Transactions on
  Wireless Communications}, vol.~15, pp.~5087--5100, July 2016.

\bibitem{su_fractional_2015}
L.~Su and C.~Yang, ``Fractional frequency reuse aided pilot decontamination for
  massive {MIMO} systems,'' in {\em {IEEE} 81st {Vehicular} {Technology}
  {Conference} ({VTC} {Spring})}, pp.~1--6, May 2015.

\bibitem{yan_pilot_2015}
X.~Yan, H.~Yin, M.~Xia, and G.~Wei, ``Pilot sequences allocation in {TDD}
  massive {MIMO} systems,'' in {\em {IEEE} {Wireless} {Communications} and
  {Networking} {Conference} ({WCNC})}, pp.~1488--1493, March 2015.

\bibitem{zhu_smart_2015}
X.~Zhu, Z.~Wang, L.~Dai, and C.~Qian, ``Smart {pilot} {assignment} for
  {massive} {MIMO},'' {\em IEEE Communications Letters}, vol.~19,
  pp.~1644--1647, Sept. 2015.

\bibitem{fernandes_inter-cell_2013}
F.~Fernandes, A.~Ashikhmin, and T.~L. Marzetta, ``Inter-{cell} {interference}
  in {noncooperative} {TDD} {large} {scale} {antenna} {systems},'' {\em IEEE
  Journal on Selected Areas in Communications}, vol.~31, pp.~192--201, Feb.
  2013.

\bibitem{lu2014overview}
L.~Lu, G.~Y. Li, A.~L. Swindlehurst, A.~Ashikhmin, and R.~Zhang, ``An overview
  of massive mimo: Benefits and challenges,'' {\em IEEE Journal of Selected
  Topics in Signal Processing}, vol.~8, no.~5, pp.~742--758, 2014.

\bibitem{khormuji_pilot-decontamination_2016}
M.~N. Khormuji, ``Pilot-decontamination in massive {MIMO} systems via network
  pilot-data alignment,'' in {\em Communications {Workshops} ({ICC}), 2016
  {IEEE} {International} {Conference} on}, pp.~93--97, IEEE, 2016.

\bibitem{saxena_mitigating_2015}
V.~Saxena, G.~Fodor, and E.~Karipidis, ``Mitigating pilot contamination by
  pilot reuse and power control schemes for massive {MIMO} systems,'' in {\em
  2015 {IEEE} 81st {Vehicular} {Technology} {Conference} ({VTC} {Spring})},
  pp.~1--6, IEEE, 2015.

\bibitem{pang_optimal_2007}
J.~Pang, J.~Li, L.~Zhao, and Z.~Lu, ``Optimal training sequences for {MIMO}
  channel estimation with spatial correlation,'' in {\em {IEEE} {Vehicular}
  {Technology} {Conference} (VTC)}, pp.~651--655, Sept 2007.

\bibitem{liu_training_2007}
Y.~Liu, T.~F. Wong, and W.~W. Hager, ``Training {signal} {design} for
  {estimation} of {correlated} {MIMO} {channels} {with} {colored}
  {interference},'' {\em IEEE Transactions on Signal Processing}, vol.~55,
  pp.~1486--1497, Apr. 2007.

\bibitem{kotecha_transmit_2004}
J.~Kotecha and A.~Sayeed, ``Transmit {signal} {design} for {optimal}
  {estimation} of {correlated} {MIMO} {channels},'' {\em IEEE Transactions on
  Signal Processing}, vol.~52, pp.~546--557, Feb. 2004.

\bibitem{noh_pilot_2014}
S.~Noh, M.~D. Zoltowski, Y.~Sung, and D.~J. Love, ``Pilot {beam} {pattern}
  {design} for {channel} {estimation} in {massive} {MIMO} {systems},'' {\em
  IEEE Journal of Selected Topics in Signal Processing}, vol.~8, pp.~787--801,
  Oct. 2014.

\bibitem{bogale_pilot_2014}
T.~E. Bogale and L.~B. Le, ``Pilot optimization and channel estimation for
  multiuser massive {MIMO} systems,'' in {\em {Conference} on Information
  {Sciences} and {Systems} ({CISS})}, pp.~1--6, Mar. 2014.

\bibitem{bjornson_framework_2010}
E.~Bjornson and B.~Ottersten, ``A {framework} for {training}-{based}
  {estimation} in {arbitrarily} {correlated} {Rician} {MIMO} {channels} {with}
  {Rician} {disturbance},'' {\em IEEE Transactions on Signal Processing},
  vol.~58, pp.~1807--1820, Mar. 2010.

\bibitem{alkhateeb_channel_2014}
A.~Alkhateeb, O.~El~Ayach, G.~Leus, and R.~W. Heath, ``Channel {estimation} and
  {hybrid} {precoding} for {millimeter} {wave} {cellular} {systems},'' {\em
  IEEE Journal of Selected Topics in Signal Processing}, vol.~8, pp.~831--846,
  Oct. 2014.

\bibitem{ayach_spatially_2014}
O.~E. Ayach, S.~Rajagopal, S.~Abu-Surra, Z.~Pi, and R.~W. Heath, ``Spatially
  {sparse} {precoding} in {millimeter} {wave} {MIMO} {systems},'' {\em IEEE
  Transactions on Wireless Communications}, vol.~13, pp.~1499--1513, Mar. 2014.

\bibitem{li_hybrid_2017}
N.~Li, Z.~Wei, H.~Yang, X.~Zhang, and D.~Yang, ``Hybrid {precoding} for
  {mmWave} {massive} {MIMO} {systems} {with} {partially} {connected}
  {structure},'' {\em IEEE Access}, vol.~5, pp.~15142--15151, 2017.

\bibitem{mendez-rial_hybrid_2016}
R.~Mendez-Rial, C.~Rusu, N.~Gonzalez-Prelcic, A.~Alkhateeb, and R.~W. Heath,
  ``Hybrid {MIMO} {architectures} for {millimeter} {wave} {communications}:
  {phase} {shifters} or {switches}?,'' {\em IEEE Access}, vol.~4, pp.~247--267,
  2016.

\bibitem{yu_alternating_2016}
X.~Yu, J.-C. Shen, J.~Zhang, and K.~B. Letaief, ``Alternating {minimization}
  {algorithms} for {hybrid} {precoding} in {millimeter} {wave} {MIMO}
  {systems},'' {\em IEEE Journal of Selected Topics in Signal Processing},
  vol.~10, pp.~485--500, Apr. 2016.

\bibitem{yu_partially-connected_2017}
X.~Yu, J.~Zhang, and K.~B. Letaief, ``Partially-{connected} {hybrid}
  {precoding} in mm-{Wave} {systems} {with} {dynamic} {phase} {shifter}
  {networks},'' {\em arXiv preprint arXiv:1705.00859}, 2017.

\bibitem{kim_mse-based_2015}
M.~Kim and Y.~H. Lee, ``{MSE}-{based} {hybrid} {RF}/{baseband} {processing} for
  {millimeter}-{wave} {communication} {systems} in {MIMO} {interference}
  {channels},'' {\em IEEE Transactions on Vehicular Technology}, vol.~64,
  pp.~2714--2720, June 2015.

\bibitem{ioushua_eldar_hybrid_2017}
S.~S. Ioushua and Y.~C. Eldar, ``Hybrid analog-digital beamforming for massive
  mimo systems,'' {\em arXiv preprint arXiv:1712.03485}, 2017.

\bibitem{kermoal_stochastic_2002}
J.~Kermoal, L.~Schumacher, K.~Pedersen, P.~Mogensen, and F.~Frederiksen, ``A
  stochastic {MIMO} radio channel model with experimental validation,'' {\em
  IEEE Journal on Selected Areas in Communications}, vol.~20, pp.~1211--1226,
  Aug. 2002.

\bibitem{tulino_impact_2005}
A.~Tulino, A.~Lozano, and S.~Verdu, ``Impact of {antenna} {correlation} on the
  {capacity} of {multiantenna} {channels},'' {\em IEEE Transactions on
  Information Theory}, vol.~51, pp.~2491--2509, July 2005.

\bibitem{Yu01secondorder}
K.~Yu, M.~Bengtsson, B.~Ottersten, D.~Mcnamara, P.~Karlsson, and M.~Beach,
  ``Second order statistics of nlos indoor mimo channels based on 5.2 ghz
  measurements,'' in {\em IEEE Global Telecommunications Conference
  (GLOBECOM)}, pp.~156--160, Nov. 2001.

\bibitem{Kay:1993:FSS:151045}
S.~M. Kay, {\em Fundamentals of statistical signal processing: estimation
  theory}.
\newblock Upper Saddle River, NJ, USA: Prentice-Hall, Inc., 1993.

\bibitem{spawc_paper}
S.~S. Ioushua and Y.~C. Eldar, ``Pilot contamination mitigation with reduced
  {RF} chains,'' {\em Signal Processing Advances in Wireless Communications
  (SPAWC)}, July 2017.

\bibitem{marshall_inequalities:_????}
A.~W. Marshall, I.~Olkin, and B.~C. Arnold, {\em Inequalities: {theory} of
  {majorization} and {its} {application}}.
\newblock Springer series in statistics, New York: Springer, 2011.

\bibitem{gesbert_shifting_2007}
D.~Gesbert, M.~Kountouris, R.~Heath~Jr., C.-b. Chae, and T.~Salzer, ``Shifting
  the {MIMO} {Paradigm},'' {\em IEEE Signal Processing Magazine}, vol.~24,
  pp.~36--46, Sept. 2007.

\bibitem{shirani-mehr_joint_2011}
H.~Shirani-Mehr, H.~Papadopoulos, S.~A. Ramprashad, and G.~Caire, ``Joint
  {Scheduling} and {ARQ} for {MU}-{MIMO} {Downlink} in the {Presence} of
  {Inter}-{Cell} {Interference},'' {\em IEEE Transactions on Communications},
  vol.~59, pp.~578--589, Feb. 2011.

\bibitem{kwan_proportional_2009}
R.~Kwan, C.~Leung, and J.~Zhang, ``Proportional {Fair} {Multiuser} {Scheduling}
  in {LTE},'' {\em IEEE Signal Processing Letters}, vol.~16, pp.~461--464, June
  2009.

\bibitem{viswanath_opportunistic_2002}
P.~Viswanath, D.~N.~C. Tse, and R.~Laroia, ``Opportunistic beamforming using
  dumb antennas,'' {\em IEEE transactions on information theory}, vol.~48,
  no.~6, pp.~1277--1294, 2002.

\bibitem{huh_multi-cell_2011}
H.~Huh, S.-H. Moon, Y.-T. Kim, I.~Lee, and G.~Caire, ``Multi-{Cell} {MIMO}
  {Downlink} {With} {Cell} {Cooperation} and {Fair} {Scheduling}: {A}
  {Large}-{System} {Limit} {Analysis},'' {\em IEEE Transactions on Information
  Theory}, vol.~57, pp.~7771--7786, Dec. 2011.

\bibitem{abhyankar_min-max_2007}
A.~R. Abhyankar, S.~A. Soman, and S.~A. Khaparde, ``Min-{Max} {Fairness}
  {Criteria} for {Transmission} {Fixed} {Cost} {Allocation},'' {\em IEEE
  Transactions on Power Systems}, vol.~22, pp.~2094--2104, Nov. 2007.

\bibitem{shen_solving_2009}
P.~Shen, Y.~Chen, and Y.~Ma, ``Solving sum of quadratic ratios fractional
  programs via monotonic function,'' {\em Applied Mathematics and Computation},
  vol.~212, pp.~234--244, June 2009.

\bibitem{shen_maximizing_2013}
P.~Shen, W.~Li, and X.~Bai, ``Maximizing for the sum of ratios of two convex
  functions over a convex set,'' {\em Computers \& Operations Research},
  vol.~40, pp.~2301--2307, Oct. 2013.

\bibitem{jaberipour_solving_2010}
M.~Jaberipour and E.~Khorram, ``Solving the sum-of-ratios problems by a harmony
  search algorithm,'' {\em Journal of Computational and Applied Mathematics},
  vol.~234, pp.~733--742, June 2010.

\bibitem{Petersen06thematrix}
K.~B. Petersen, M.~S. Pedersen, J.~Larsen, K.~Strimmer, L.~Christiansen,
  K.~Hansen, L.~He, L.~Thibaut, M.~Barão, S.~Hattinger, V.~Sima, and W.~The,
  ``The matrix cookbook,'' tech. rep., 2006.

\end{thebibliography}
\end{document}